%% file: arxiv.tex
\pdfoutput=1 
\documentclass[twocolumn]{article}
\usepackage[american]{babel}

\usepackage[paper=a4paper,margin=1.5cm]{geometry}
\usepackage{hyperref}
\usepackage{natbib} 
\usepackage{mathtools} 
\usepackage{booktabs} 
\usepackage{tikz} 
\usepackage{amssymb}
\usepackage{amsthm}
\usepackage[ruled,vlined,linesnumbered]{algorithm2e}
\usepackage{pgfplots}	
\pgfplotsset{compat=1.17}
\usetikzlibrary{arrows}
\usepgfplotslibrary{fillbetween}
\usepackage{subcaption}

\input{sections/custom_commands.tex}
\theoremstyle{plain}
\newcounter{thm}
\newtheorem{theorem}[thm]{Theorem}
\newtheorem{lemma}[thm]{Lemma}
\newtheorem{corollary}[thm]{Corollary}
\theoremstyle{remark}
\newtheorem*{remark}{Remark}

\title{Inference of Causal Effects when Control Variables are Unknown}
\author{ Ludvig Hult\\Uppsala University \and Dave Zachariah\\Uppsala University }
\date{}

\begin{document}
\maketitle

\input{sections/abstract}
\input{sections/introduction}

\input{sections/problem_formulation}

\input{sections/results}
\input{sections/numerics}
\input{sections/conclusion}

\subsubsection*{Contributions}
    Ludvig~Hult made the numerical simulations, the theoretical derivations and typeset the technical parts as well as produced all figueres and diagram.
    All code is due to Ludvig~Hult.

    Dave~Zachariah concieved the idea, guided the work and supported the article authoring.

\subsubsection*{Acknowledgements}
    This work was partly supported by the Swedish Research Council under contract 2018-05040  and the \emph{Wallenberg AI, Autonomous Systems and Software Program (WASP)} funded by Knut and Alice Wallenberg Foundation.

\bibliography{hult_502}

\clearpage
\onecolumn
\input{sections/suppl_mat}

\end{document}

%% file: sections/custom_commands.tex
\newcommand{\R}{\mathbb R}

\newcommand{\kroneckerDelta}{\delta}
\DeclarePairedDelimiter\norm{\lVert}{\rVert}
\DeclareMathOperator*{\vecop}{vec}
\DeclareMathOperator*{\diag}{diag}
\DeclareMathOperator*{\matop}{mat}
\newcommand{\eye}{I}
\DeclareMathOperator*{\tr}{tr}
\newcommand{\hadamard}{\circ}
\newcommand{\kronecker}{\otimes}
\newcommand{\T}{\top}
\newcommand{\vecOne}{\mathbf{1}}

\DeclareMathOperator*{\cov}{Cov}
\DeclareMathOperator*{\argmin}{arg\,min}
\newcommand{\convp}{\overset{p}{\to}}
\newcommand{\convd}{\overset{d}{\to}}
\newcommand{\E}{\mathbb{E}}
\newcommand{\En}{\mathbb{E}_n}
\newcommand{\Eint}{\widetilde{\mathbb{E}}}
\newcommand{\covint}{\widetilde{\text{Cov}}}
\newcommand{\varint}{\widetilde{\text{Var}}}
\newcommand{\var}{\text{Var}}
\newcommand{\Prob}{\mathbb{P}}
\newcommand{\normal}{\mathcal N}

\newcommand{\unitBasisMatrix}{E}

\newcommand{\semCoeffMat}{W}

\newcommand{\semCoeffMatSet}{\mathcal W}
\newcommand{\semCoeffOpt}{\semCoeffMat_\circ}

\newcommand{\semCoeffEstN}{\semCoeffMat_\nData}

\newcommand{\semScaleMatrix}{M}
\newcommand{\semVector}{v}

\newcommand{\semNoise}{e}
\newcommand{\semNoiseCovariance}{\Sigma}
\newcommand{\semInterventionNoise}{\widetilde \semNoise}
\newcommand{\semInterventionNoiseCovariance}{\widetilde \semNoiseCovariance}
\newcommand{\mutilatingMatrix}{Z}
\newcommand{\cofactorMatrix}{C}
\newcommand{\dNodes}{d}

\newcommand{\averageCausalEffect}{\gamma}
\newcommand{\averageCausalEffectTarget}{\averageCausalEffect_\circ}
\newcommand{\averageCausalEffectSet}{\Gamma}
\newcommand{\averageCausalEffectEstN}{\averageCausalEffect_\nData}
\newcommand{\averageCausalEffectNumeric}{\hat{\averageCausalEffectTarget}}
\newcommand{\observationalDistribution}{p}
\newcommand{\interventionalDistribution}{\tilde p}

\newcommand{\outcomeVar}{y}
\newcommand{\decisionVar}{x}
\newcommand{\adjustmentVar}{z}
\newcommand{\adjusted}[1]{\bar{#1}}
\newcommand{\validAdjustmentVar}{\adjusted{z}}

\newcommand{\nData}{n}
\newcommand{\dagTolerance}{\epsilon}
\newcommand{\dagToleranceMax}{\dagTolerance_\star}
\newcommand{\mEstParameter}{\theta}
\newcommand{\mEstParameterEstN}{\mEstParameter_\nData}
\newcommand{\mEstParameterTrue}{\mEstParameter_\circ}
\newcommand{\mEstParameterSet}{\Theta}
\newcommand{\mEstParametrization}{L}
\newcommand{\mEstLoss}{\ell}
\newcommand{\mEstConstrint}{g}
\newcommand{\mEstCovarianceN}{\mathcal J_\nData}
\newcommand{\regCoefficient}{\beta}
\newcommand{\regCoefficientSet}{B}

\newcommand{\confidenceLevel}{\alpha}
\newcommand{\hFun}{h}
\newcommand{\qMatrix}{\mathsf{Q}}

\newcommand{\Un}{U_n}
\newcommand{\Qn}{Q_n}
\newcommand{\Qtrue}{Q_\circ}
\newcommand{\Utrue}{U_\circ}
\newcommand{\Jn}{J_n}
\newcommand{\Kn}{K_n}
\newcommand{\Ktrue}{K_\circ}
\newcommand{\Jtrue}{J_\circ}
\newcommand{\nablatheta}{\nabla}
\newcommand{\PiTrue}{\Pi_\circ}
\newcommand{\PiN}{\Pi_n}

\newcommand{\augLagSlack}{s}
\newcommand{\augLag}{\mathcal L}
\newcommand{\augLagLagMul}{\alpha}
\newcommand{\augLagPen}{\rho}
\newcommand{\augLagPenMul}{\mu}
\newcommand{\augLagPenMax}{\augLagPen_{max}}
\newcommand{\augLagIter}{k}
\newcommand{\augLagMinImprovement}{g}
\newcommand{\augLagContraint}{c}
\newcommand{\augLagConstraintTol}{\eta}

\newcommand{\assStruct}{\widehat{\semNoiseCovariance}} 
\newcommand{\semNoiseScale}{s}
\newcommand{\misspecCond}{\kappa\left(\assStruct^{-1}\semNoiseCovariance\right)}

\newcommand{\DAG}{\textsc{dag}}
\newcommand{\scm}{\textsc{scm}}
\newcommand{\OLS}{\textsc{ols}}

%% file: sections/abstract.tex
\begin{abstract}
    Conventional methods in causal effect inference typically rely on specifying a valid set of control variables. When this set is unknown or misspecified, inferences will be erroneous.
    We propose a method for inferring average causal effects when all potential confounders are observed, but the control variables are unknown.
    When the data-generating process belongs to the class of acyclical linear structural causal models, we prove that the method yields asymptotically valid confidence intervals.
    Our results build upon a smooth characterization of linear directed acyclic graphs.
    We verify the capability of the method to produce valid confidence intervals for average causal effects using synthetic data, even when the appropriate specification of control variables is unknown.
\end{abstract}

%% file: sections/introduction.tex
\section{Introduction}
When applied researchers aim to assess the causal effect of some policy or exposure, they must often infer it from observational data. This requires controlling for variations in the outcome of interest that arise from confounding factors. After selecting a set of control variables, inferences are often drawn using regression models. But selecting a valid control variable set is in general hard and the use of invalid sets produces misleading inferences, see. e.g., \citet{carlson_illusion_2012,berneth_acritical_2016}. It is therefore of practical interest to infer causal effects without relying on the researcher to specify the control variables among all observed variables. 

In this paper, we will develop such an inferential method under the assumption that there is no unobserved confounding. The method infers average causal effects using asymptotic confidence intervals and obviates the need for specifying control variables.
 
Consider a random outcome variable $\outcomeVar$ observed after an intervention on another scalar $\decisionVar$. We denote the unknown conditional distribution of outcomes under such an intervention as
\[\outcomeVar \sim \interventionalDistribution(\outcomeVar | \decisionVar)\]
We consider the scalars $\decisionVar$ and $\outcomeVar$ to be of zero mean, i.e. $\Eint[x]=\Eint[y]=0$, where the tilde denotes that the expectation is taken with respect to the interventional distribution $\interventionalDistribution$.
The conditional mean function $\Eint[\outcomeVar|\decisionVar]$ describes the effect of the intervention and can be summarized by the distribution parameter
\begin{align}
\boxed{\averageCausalEffect \coloneqq \: \frac{\covint[\decisionVar,  \outcomeVar]}{\varint[\decisionVar]} \; \equiv \;  \argmin_{\bar{\averageCausalEffect}} \; \Eint\left[ \big( \Eint[\outcomeVar|\decisionVar]  - \bar{\averageCausalEffect}\decisionVar \big)^2 \right]}
\end{align}
Thus $\averageCausalEffect \decisionVar$ is an optimal linear approximation of the conditional mean function. When the conditional mean function is linear, the parameter is the average causal effect of the intervention, i.e., $\averageCausalEffect \equiv \frac{\partial}{\partial \decisionVar } \Eint[\outcomeVar|\decisionVar]$  \citep{angrist_mostly_2009,pearl_causality:_2009}.

The task is to infer $\averageCausalEffect$ using data from a different, \emph{observational} distribution
\begin{align}
\label{eq:dataGeneratingProcess}    
 (\decisionVar_i, \outcomeVar_i, \adjustmentVar_i) \sim \observationalDistribution(\decisionVar, \outcomeVar, \adjustmentVar),\quad i=1, \dots, \nData 
\end{align}
where $\adjustmentVar$ is a vector of additional random variables. A standard procedure to infer $\averageCausalEffect$ is to use the partial regression coefficient
\begin{equation} \label{eq:partial_regression_coeff}
    \begin{split}
    \regCoefficient \coloneqq  \frac{\cov[\adjusted{\decisionVar},  \adjusted{\outcomeVar}]}{\var[\adjusted{\decisionVar}]},
    \end{split}
\end{equation}
where $\adjusted{\decisionVar}$ and $\adjusted{\outcomeVar}$ are adjusted according to 
\begin{equation}
\begin{split}
     \bar{\decisionVar} &\coloneqq \decisionVar - \cov[\decisionVar,  \validAdjustmentVar]\cov[\validAdjustmentVar]^{-1}\validAdjustmentVar \\
\bar{\outcomeVar} &\coloneqq \outcomeVar - \cov[\outcomeVar,  \validAdjustmentVar]\cov[\validAdjustmentVar]^{-1}\validAdjustmentVar,
\end{split}
\end{equation}
where $\validAdjustmentVar \subseteq \adjustmentVar$ is a set of \emph{control variables} using the terminology in much of regression analysis. If this set were \emph{valid}, the noncausal association between $\decisionVar$ and  $\outcomeVar$ can be blocked. Then $\regCoefficient = \averageCausalEffect$ when the data-generating process is well-described by a linear model \citep{angrist_mostly_2009,pearl_causality:_2009}. See \citep[ch.~6.6]{peters_elements_2017} for a general definition of valid  control variables using structural causal models (\scm). Throughout the paper, we will assume that at least one valid subset of $\adjustmentVar$ exists but that it is \emph{unknown}. If a specified $\validAdjustmentVar$ contains invalid controls, the resulting inferences become erroneous as the following example illustrates.

\begin{figure*}[ht!]
     \centering
     \begin{subfigure}[b]{0.45\textwidth}
         \centering
         \tikzset{node distance=2.3cm}
         \input{tikz/4node_collider_graph.tikz}
         \caption{Underlying causal structure}
         \label{fig:collider_dag}
     \end{subfigure}
     \hfill
     \begin{subfigure}[b]{0.45\textwidth}
         \centering
        \pgfplotsset{every axis/.append style={width=\columnwidth, height=.6\columnwidth}}
         \input{tikz/4node_collider_chart.tikz}
         \caption{$95\%$-confidence intervals that aim to cover $\averageCausalEffectTarget$}
         \label{fig:collider_asymptotics}
     \end{subfigure}
     \hfill
     \caption{Using observational data \eqref{eq:dataGeneratingProcess} generated by a linear \scm{} based on (a), we aim to infer an unknown causal parameter $\averageCausalEffectTarget$ (further details in Section~\ref{subsection:correctly_identify_adjustment}).
     The causal structure is here unknown and using $\adjustmentVar=[\adjustmentVar_1 ,\, \adjustmentVar_2]$ as the control variables, the standard approach based on the ordinary least-squares (\OLS) method yields confidence interval $\regCoefficientSet_{\confidenceLevel,\nData}$ in (b). Since $\adjustmentVar$ is invalid
     due to the collider bias induced by $\adjustmentVar_1$, the inferences are erroneous.
     Below we develop an inference method that yields calibrated confidence intervals $\averageCausalEffectSet_{\confidenceLevel,\nData}$ when the causal structure in (a), and therefore a set of valid control variables, is unknown.}
\end{figure*}

\paragraph{Example: Invalid control variables} 
Consider a data-generating process with a causal structure as illustrated in Figure~\ref{fig:collider_dag}. Only $\adjustmentVar_2 \subset \adjustmentVar$ constitutes  a valid control variable, by blocking the noncausal association between $\decisionVar$ and $\outcomeVar$. Neither $\varnothing$ nor $\adjustmentVar_1$ are valid. If the causal structure is unknown or misspecified so that we use $\validAdjustmentVar = [\adjustmentVar_1, \adjustmentVar_2]^\T$ instead of $\adjustmentVar_2$, then inferring $\regCoefficient$ in equation \eqref{eq:partial_regression_coeff} will yield erroneous conclusions about the average causal effect, as shown in Figure~\ref{fig:collider_asymptotics}. We also illustrate an alternative methodology developed in this paper which, by contrast, does not require a correctly specified causal structure.

\paragraph{Contribution and related work} The contribution of this paper is the development of a confidence interval for the average causal effect that obviates the need to specify valid control variables, and we derive its statistical properties.

To decide the valid control variables among $\adjustmentVar$, typically requires the causal structure of the data-generating process. The problem of learning such structures from data, aka. causal discovery, has been studied over a few decades \citep{spirtes_causation_1993, pearl_causality:_2009,peters_elements_2017}. A central challenge of the field is to optimize model fitness over the discrete nature of graphs representing the causal structure. \citet{zheng_dags_2018} proposed a smooth characterization of directed acyclic graphs (\DAG{}) which enables conventional optimization methods to be used. See \citep{yu_daggnn_2019,ke_learning_2019,brouillard_diffrential_2020,zheng_learning_2020,kyono_2020} for applications and extentions of this methodology. 

Our method presented herein utilizes that characterization of \DAG{}s and builds upon the  framework of M-estimation. See e.g. the presentation in \citep[ch. 12]{wooldridge_econometric_2010} or \citet{vaart_m-_1998} for an introduction. 
When imposing \DAG{}-constraints, we find the need to extend the basic M-estimation framework. While the theory of constrained M-estimation has been approached before \citep{geyer_asymptotics_1994,shapiro_asymptotics_2000,andrews_estimation_1999, wang_asymptotics_1996}, we show that the assumptions needed do not hold due to the geometry of the  \DAG{} constraints. Moreover, alternative characterizations of \DAG{}s, presented in  \citet{wei_dags_2020}, would not remedy this problem.

Therefore we take a different approach, inspired by \citet{stoica_cramer-rao_1998}, to derive the large-sample properties of the proposed confidence interval and prove its asymptotic validity. Our theoretical results are corroborated by numerical experiments, which demonstrate the ability of the method to correctly infer average causal effects in linear \scm{}s without specifying valid control variables.

Lastly we emphasize that while our method builds upon insights from the causal discovery literature, its task is to infer the average causal effect and not a causal graph.

%% file: tikz/4node_collider_graph.tikz
\begin{tikzpicture}[->,>=stealth',shorten >=1pt,auto,semithick]
  \node[circle,draw, minimum size=20pt] (Z2) {$\adjustmentVar_1$};
  \node[circle,draw, minimum size=20pt] (X) [above right of=Z2] {$\decisionVar$};
  \node[circle,draw, minimum size=20pt] (Y) [below right of=Z2] {$\outcomeVar$};
  \node[circle,draw, minimum size=20pt] (Z1) [below right of=X] {$\adjustmentVar_2$};

  \path (Z1) edge  (X)
  (Z1) edge   (Y)
  (X) edge   (Z2)
  (Y) edge   (Z2);
\end{tikzpicture}

%% file: tikz/4node_collider_chart.tikz
\begin{tikzpicture}[baseline]
    \pgfplotstableread[col sep=comma]{./data/4node_collider_summary.csv}{\datatable};
    \begin{semilogxaxis}[
        xlabel={No. of data points, $\nData$},
        ylabel={Parameter $\averageCausalEffect$},
        xmin=90,
        xmax=11000,
        legend style={
            font=\scriptsize, at={(0.5,1.10)},
            anchor=south,legend columns=-1},
    ]
        \addplot+ [ only marks, mark=*, mark size=1pt,
            error bars/.cd,
            y dir=both,
            y explicit] table [x=m_obs, y=ace_value, y error=q_ace_standard_error] {\datatable};
        \addplot+ [only marks, mark=*, mark size= 1 pt,
            error bars/.cd,
            y dir=both,
            y explicit] table [x=m_obs, y=ols_value, y error=q_ols_standard_error] {\datatable};
        \addplot [no markers] table [x=m_obs, y=ace_circ] {\datatable};
        \legend{{$\averageCausalEffectSet_{\confidenceLevel,\nData}$},{$\regCoefficientSet_{\confidenceLevel,\nData}$},{$\averageCausalEffectTarget$}}
    \end{semilogxaxis}
\end{tikzpicture}

%% file: sections/problem_formulation.tex
\section{Problem Formulation}

We begin by specifying the class models for the data generating process that we will consider and then proceed to define the target quantity that we seek to infer from data.

\subsection{Model Class for the Data-Generating Process}

To simplify the notation, we introduce the $\dNodes$-dimensional data vector $\semVector^\T = (\decisionVar ,\, \outcomeVar ,\, \adjustmentVar^\T)$. Suppose the data-generating process $p(\semVector)$ in \eqref{eq:dataGeneratingProcess} belongs to the class of linear \scm. That is, we can express the data vector as
\begin{align}
    \label{eq:sem:linear_sem}
    \semVector = \semCoeffMat^\T\semVector + \semNoise,
\end{align}
where is $\semNoise$ is zero-mean random variable with a diagonal covariance matrix $\semNoiseCovariance$. It is for simplicity assumed to be known here, although as we point out in Section~\ref{sec:result} this assumption can be relaxed to a certain degree.  We let $\semCoeffMat\in \R^{\dNodes\times\dNodes}$ have zeros on its diagonal. It can be interpreted as a weighted directed graph, by letting $W_{i,j}$ be the weight on the edge from node $i$ to node $j$.
The matrix $W^\T$ is sometimes referred to as the \emph{adjacency matrix} \citep{shimizu_directlingam_2011} or the \emph{autoregressive matrix} \citep{loh_high-dimensional_2014}.

The matrix $\semCoeffMat$ is unknown but has certain restrictions. For \scm{}s it is common to impose a \DAG{} structure on the graph specified by $\semCoeffMat$, since such structure significantly clarifies and simplifies any causal analysis of the model.
We will call $\semCoeffMat$ a `\DAG{}-matrix' if the directed graph of the matrix is acyclical.
When $\semCoeffMat$ is a \DAG{}-matrix, we can interpret the entry $\semCoeffMat_{i,j}$ as the expected increase in $\semVector_i$ for every unit increase in $\semVector_j$, holding all other variables constant.

\citet{zheng_dags_2018} introduced the function $h(W) \coloneqq \tr \exp (\semCoeffMat \hadamard \semCoeffMat)-\dNodes$, using the trace of the matrix exponential and the element-wise product $\hadamard$, and showed that
\[ \semCoeffMat \text{ is \DAG-matrix} \Leftrightarrow \hFun(\semCoeffMat) =0\]
To enable a tractable analysis below, we will also consider the set of all $\dagTolerance$-almost \DAG{}-matrices, defined as
\begin{align}
    \label{eq:def:semCoeffMatSet}
    \semCoeffMatSet_\dagTolerance = \left\{ \semCoeffMat\, \middle| \hFun(\semCoeffMat) \leq \dagTolerance  \text{ and } \diag(\semCoeffMat)=0 \right\}
\end{align}
Note that when $\dagTolerance=0$, the set  $\semCoeffMatSet_0$ is exactly the set of \DAG-matrices.
When $\dagTolerance > 0$, cycles are permitted but the magnitude of their effects are bounded. Below we will provide bounds on $\dagTolerance$ that enable a meaningful analysis of $\semCoeffMatSet_\dagTolerance$.

Given the data-generating process in \eqref{eq:sem:linear_sem}, we can define an \emph{interventional} distribution $\widetilde{p}(\semVector)$ with respect to the first variable $\decisionVar$ \citep{pearl_causality:_2009}:
Introduce $\mutilatingMatrix$, a matrix with ones on the diagonal, except the first element, which is zero, i.e.
\begin{equation}\label{eq:def:mutilatingMatrix}
    Z \in \R^{\dNodes\times\dNodes},\quad \mutilatingMatrix_{i,j} = \begin{cases} 1& \text{if } i=j>1\\ 0 &\text{else}\end{cases}
\end{equation}
Next, introduce a new random vector $\semInterventionNoise $, with the same statistical properties as $\semNoise$ in \eqref{eq:sem:linear_sem} for all components, but for its first component, and let $\semInterventionNoiseCovariance$ denote its diagonal covariance matrix.
The interventional distribution $\widetilde{p}(\semVector)$ is then specified by the model
\begin{align}
    \label{eq:mutilated_dgp}
    \semVector = \mutilatingMatrix \semCoeffMat^\T \semVector +  \semInterventionNoise ,
\end{align}
assuming that $(\eye - \mutilatingMatrix \semCoeffMat^\T )$ is full rank.

\subsection{Target Quantity}

For an interventional distribution given by \eqref{eq:mutilated_dgp}, we observe the following result.
\begin{lemma}
    \label{lemma:averageCausalEffectInSem}
    The average causal effect of $\decisionVar$ on $\outcomeVar$ in a linear \scm{} with interventional distribution $\widetilde{p}(\semVector)$ is
    \begin{equation}
        \label{eq:def:averageCausalEffectInSem}
        \averageCausalEffect(\semCoeffMat) = \frac{\covint[\decisionVar,  \outcomeVar]}{\varint[\decisionVar]} \equiv \left[ (\eye - \mutilatingMatrix \semCoeffMat^\T )^{-1} \right]_{2,1}
    \end{equation}
    where $\semCoeffMat$ is a (possibly non-\DAG{}) adjacency matrix.
\end{lemma}
The syntax $[.]_{2,1}$ refers to the second row and first column of a matrix. The proof is a direct computation and given in the supplementary material.

We are interested in computing the average causal effect
\begin{subequations}
    \begin{equation}\label{eq:def:averageCausalEffectTarget}
        \boxed{\averageCausalEffectTarget = \averageCausalEffect(\semCoeffOpt),}
    \end{equation}
    where $\semCoeffOpt$ is an $\dagTolerance$-almost \DAG{} adjacency matrix that optimally fits the observational data using the following criterion,
    \begin{equation}\label{eq:def:semcoefftrue}
        \semCoeffOpt \coloneqq \argmin_{\semCoeffMat \in \semCoeffMatSet_\dagTolerance} \; \E\Big[  \norm{\semNoiseCovariance^{-1/2}  (\eye - \semCoeffMat^\T)\semVector }^2 \Big]
    \end{equation}
\end{subequations}
\citet[corollary~8]{loh_high-dimensional_2014} show that if the observational distribution $p(\adjustmentVar)$ follows \eqref{eq:sem:linear_sem} and $\dagTolerance=0$, then \eqref{eq:def:semcoefftrue} correctly identifies the unknown matrix. Moreover, \citet[theorem~9]{loh_high-dimensional_2014} proves that identifiability is obtained even under limited misspecification of the entries in $\cov[\semNoise] = \semNoiseCovariance$. Thus the target quantity $\averageCausalEffectTarget$ is defined as the average causal effect of the optimally fitted linear \scm{} and requires no further distributional assumptions.

Our task is to construct a confidence interval $\averageCausalEffectSet_{\alpha,\nData}$, that is using $\nData$ data points, and has a coverage probability $1-\confidenceLevel$ for the quantity $\averageCausalEffectTarget$.

%% file: sections/results.tex
\section{Results} 
\label{sec:result}

We present the results in this paper in two parts. First, we present the confidence interval for $\averageCausalEffectTarget$ with an asymptotically valid coverage probability (Theorem~\ref{thm:cofidence_interval_for_ace}). This uses a general result of equality-constrained M-estimation, which we subsequently present (Theorem~\ref{thm:constrained_m_estimation}, Corollary~\ref{cor:constrained_m_est}).

\subsection{Derivation of Confidence Interval}

Using the empirical average operator $\En$, we define the empirical analog of \eqref{eq:def:semcoefftrue} as
\begin{align}
    \label{eq:def:Wn}
    \semCoeffEstN \coloneqq \argmin_{\semCoeffMat \in \semCoeffMatSet_\dagTolerance} \;
    \En\Big[ \norm{\semNoiseCovariance^{-1/2}  \left(\eye - \semCoeffMat^\T \right)\semVector }^2 \Big]
\end{align}
Using $\semCoeffEstN$ and \eqref{eq:def:averageCausalEffectInSem} yields a point estimate of $\averageCausalEffectTarget$:
\begin{align}
    \label{eq:def:gamma_n}
    \averageCausalEffectEstN  \coloneqq \averageCausalEffect(\semCoeffEstN)
\end{align}

For notational simplicity, we reparameterize $\semCoeffMat$, which contains zeros along the diagonal,
by $\vecop(\semCoeffMat) = \mEstParametrization \mEstParameter$, where $L$ is a $\dNodes^2  \times \dNodes(\dNodes-1)$ matrix constructed using a $\dNodes^2  \times \dNodes^2$ identity matrix removing columns $d(k-1)+k$ for $k=1,2,\dots, d$. Using this parametrization, we formulate the loss function
\begin{align}
    \label{eq:def:mest:loss}
    \footnotesize{
        \mEstLoss_\mEstParameter(\semVector) \coloneqq (\mEstParametrization\mEstParameter - \vecop(I))^\T
        \left[\semNoiseCovariance^{-1} \kronecker  [\semVector\semVector^\T] \right]
        (\mEstParametrization\mEstParameter - \vecop(I))
    }
\end{align}
using the Kronecker product $\kronecker$, and we write
\begin{align}
    \label{eq:def:theta0}
    \mEstParameterTrue & = \argmin_{
        h(\matop(\mEstParametrization\mEstParameter)) \leq \dagTolerance }
    \E[\mEstLoss_{\mEstParameter}(\semVector) ] \\
    \label{eq:def:thetan}
    \mEstParameterEstN & = \argmin_{
        h(\matop(\mEstParametrization\mEstParameter)) \leq \dagTolerance }
    \En[\mEstLoss_{\mEstParameter}(\semVector) ]
\end{align}
equivalently to \eqref{eq:def:semcoefftrue} and $\eqref{eq:def:Wn}$.


While setting $\dagTolerance=0$ yields exact \DAG-matrices, it also renders the problem ill-suited for inference. The set $\semCoeffMatSet_0$ is nonconvex, has  an empty interior, and constraint qualification does not hold (see Lemma~\ref{lemma:nonconvexWset} in the supplementary material). Therefore, convex optimization methods, barrier methods, and any method based on first-order optimality will be invalid. Asymptotic analysis of M-estimation typically requires convexity of the tangent cone at the optimum, and that the optimal point is stationary even under the unconstrained formulation \citep{geyer_asymptotics_1994,shapiro_asymptotics_2000}, but neither of these assumptions are fulfilled at most points in the set $\semCoeffMatSet_0$. To provide a tractable analysis, we consider $\dagTolerance>0$ below and expect almost-identification when $\dagTolerance$ is small. We start with a technical lemma.
\begin{lemma}
    \label{lemma:unconstrained_minimization}
    The minimizer $\mEstParameterTrue$ in \eqref{eq:def:theta0} is bounded. If it is also unique, then there is a value of $\dagToleranceMax$ such that the minimum is obtained at the boundary $\hFun(\matop(\mEstParametrization\mEstParameterTrue)) = \dagTolerance$ for all $\dagTolerance < \dagToleranceMax$.
\end{lemma}
\begin{proof}
    First, assume that the mimimizer of \eqref{eq:def:theta0} is not bounded. In that case, there is a sequence of feasible points $t_n$ such that $\norm{t_n} \to \infty$, and $\E [\mEstLoss_{t_n}(\semVector)] $ is decreasing. This is not possible, since $\mEstLoss_t(\semVector)$ is a positive definite quadratic in $t$. We have established the boundedness $\norm{\mEstParameterTrue} < B$, for some $B$.

    Let $\qMatrix = \semNoiseCovariance^{-1} \kronecker  \E[\semVector\semVector^\T]$, i.e. a Kronecker product of two positive definite matrices and it follows that $\qMatrix$ is positive definite. Then the objective function of \eqref{eq:def:theta0} is a positive definite quadratic with a global minimum given by the stationary point  $ \mEstParameter_\star \eqqcolon (\qMatrix{}^{1/2} L )^\dagger \qMatrix{}^{1/2} \vecop(\eye ) $ where $^\dagger$ denotes the Moore-Penrose inverse.
    When $\dagTolerance = \infty$, then $\mEstParameter_\star$ is a feasible point to the minimization problem in \eqref{eq:def:theta0}.

    Define $\dagToleranceMax = \hFun( \matop(L\mEstParameter_\star))$ and consider \eqref{eq:def:theta0} for any $\dagTolerance \in (0,\,\dagToleranceMax)$.
    Observe that $\left\{ \mEstParameter \,\middle|\, \norm{\mEstParameter}\leq B \text{ and } \hFun( \matop(L\mEstParameter)) \leq \dagTolerance \right\}$ is compact, the objective function has no stationary points on the feasible set, and $\norm{\mEstParameterTrue}< B$. Conclude that $h(\matop(\mEstParametrization\mEstParameterTrue)) = \dagTolerance$.
\end{proof}

\begin{lemma}
    \label{lemma:asymptotic_normal_mestparam}
    Assume the solution to \eqref{eq:def:theta0} is unique, and that $\dagTolerance < \dagToleranceMax$ as in Lemma~\ref{lemma:unconstrained_minimization}.
    Then the asymptotic distribution of $\mEstParameterEstN$ can be described by
    \begin{align}
        \label{eq:asymptotics_for_mEstParamEst}
        \sqrt{\nData}\mEstCovarianceN^{-1/2}(\mEstParameterEstN - \mEstParameterTrue) \convd \normal (0,\eye)
    \end{align}
    The estimated covariance of the estimator is defined as $\mEstCovarianceN = \Kn^{-1}\PiN\Jn\PiN\Kn^{-1}$, where
    $ \Kn = \mEstParametrization^\T   \left[ \semNoiseCovariance^{-1} \kronecker \En \left[  \semVector \semVector^\T \right] \right]\mEstParametrization$, $\PiN$  is a projection matrix with respect to the orthogonal complement of $ \nabla_{\mEstParameter} \hFun(\matop(L\mEstParameterEstN))$
    and  $\Jn = \mEstParametrization^\T \tilde \Jn \mEstParametrization $.

    We may compute $\PiN = \eye - (qq^\T)/(q^{\T}q)$ and $q = \mEstParametrization^\T\vecop(2\semCoeffEstN \hadamard (\exp [ \semCoeffEstN\hadamard \semCoeffEstN ])^\T)$.
    Furthermore, the matrix $\tilde \Jn$ has the expression
    \begin{multline}
        \label{eq:mest:score_variance}
        (\tilde \Jn )_{d(j-1)+i,d(l-1)+k} =
        \sum_{q,r,o,p=1}^{\dNodes} \Big\{
        \big( \En\left[ \semVector_i \semVector_q \semVector_o \semVector_k \right]- \\
        \En\left[ \semVector_i \semVector_q \right]\En\left[ \semVector_o \semVector_k \right]\big)
        \semNoiseCovariance^{-1}_{j,r}
        \semNoiseCovariance^{-1}_{p,l}
        (\semCoeffMat-\eye)_{q,r}
        (\semCoeffMat-\eye)_{o,p}
        \Big\}
    \end{multline}
\end{lemma}
\begin{proof}
    By consistency of M-estimation, \eqref{eq:def:thetan} will be a consistent estimator for \eqref{eq:def:theta0}. Adding the redundant $\norm{\mEstParameter} \leq B$-constraint in Lemma~\ref{lemma:unconstrained_minimization} makes the feasible set compact and thus fulfills the technical conditions \citep[Theorem 12.2]{wooldridge_econometric_2010}.

    By Lemma~\ref{lemma:unconstrained_minimization}, we know that the minimum will be obtained at the boundary, in the limit $\nData \to \infty$. We can therefore impose equality constraints in the minimization:
    \begin{align}
        \mEstParameterEstN = \argmin_{
            h(\matop(\mEstParametrization\mEstParameter)) = \dagTolerance }
        \En[\mEstLoss_{\mEstParameter}(\semVector) ]
    \end{align}
    Now apply Corollary~\ref{cor:constrained_m_est} derived below. It states the formula for confidence intervals under equality-constrained M-estimation using plug-in estimators of data covariance and cross-moments. The derivation of the expressions for $\tilde \Jn$, $\Kn$ and $\PiN$ from \eqref{eq:def:mest:loss} are direct computations presented in the supplementary material as Lemma~\ref{lemma:mest:symbols}. Technical conditions are presented in Lemma~\ref{lemma:mest:technicalities}.
\end{proof}

We can now state our main result for inferring the average causal effect $\averageCausalEffectTarget$.
\begin{theorem}
    \label{thm:cofidence_interval_for_ace}
    The confidence interval
    \begin{align}
        \label{eq:confidence_set_for_ace}
        \averageCausalEffectSet_{\confidenceLevel, \nData}  = \left\{ \averageCausalEffect \in \R \middle|\frac{1}{\nData} \frac{ (\averageCausalEffect - \averageCausalEffectEstN)^2}{  \nabla \averageCausalEffect(\mEstParameterEstN)^\T\mEstCovarianceN \nabla \averageCausalEffect(\mEstParameterEstN))}  \leq \chi^2_{1,\confidenceLevel}  \right\}
        \quad
    \end{align}
    has asymptotic coverage probability
    \begin{equation}
        \lim_{\nData \rightarrow \infty} \: \Prob ( \averageCausalEffectTarget \in  \averageCausalEffectSet_{\alpha, \nData}  ) = 1 - \alpha,
    \end{equation}
    where $\chi^2_{1,\confidenceLevel}$ denotes the $(1-\alpha)$ quantile of the chi-squared distribution with 1 degree of freedom.
    \label{thm:confidence_set_for_ace}
\end{theorem}
\begin{proof}

    Define $\averageCausalEffect(\mEstParameter)$ as the value of $\averageCausalEffect(\matop(\mEstParametrization \mEstParameter))$ in \eqref{eq:def:averageCausalEffectInSem}.

    The gradient $\nabla \averageCausalEffect(\mEstParameterEstN)$ may be computed on closed form by differentiating \eqref{eq:def:averageCausalEffectInSem}, obtaining
    \begin{align}
        \label{eq:gradient_of_ace}
        \left[ \nabla _\mEstParameter \averageCausalEffect(\mEstParameter) \right]_k = -\left( \left[\semScaleMatrix \mutilatingMatrix \kronecker \eye \right] \mEstParametrization \right)_{d+1,k}
    \end{align}
    where $\semScaleMatrix = (\eye-\mutilatingMatrix \semCoeffMat)^{-1}$. The computation is mostly keeping track of indices, and presented in supplementary materials as Lemma~\ref{lemma:gradient_of_ace}.
    Using the delta method with equation \eqref{eq:gradient_of_ace} together with Lemma~\ref{lemma:asymptotic_normal_mestparam}, we establish asymptotic normality. Form the Wald statistic for $\averageCausalEffectEstN$, and we may finally define a confidence interval $\averageCausalEffectSet_{\alpha, \nData}$.
\end{proof}

\subsection{M-estimation Asymptotics under Equality Constraints}\label{subsection:mestimation_with_constraints}

Next we derive a general result for the asymptotics of of equality-constrained M-estimation. The key observation is borrowed from \citet{stoica_cramer-rao_1998}: that we can project onto the (generalized) score onto the active constraints. We apply this insight to the more general M-estimation framework and derive complete asymptotic distribution of equality-constrained M-estimators.

In this section~\ref{subsection:mestimation_with_constraints} the function $\mEstLoss$ is not necessarily the same function as defined in \eqref{eq:def:mest:loss} but we use the same symbol to ease the mapping between the general result and its application.

\begin{theorem} \label{thm:constrained_m_estimation}
    Assume that technical conditions for consistency of M-estimation holds \citep[Theorem 12.2]{wooldridge_econometric_2010}), as well as
    \begin{itemize}
        \item The loss function $\mEstLoss_\mEstParameter(\semVector)$ is two times continously diffrentiable in $\semVector$.
        \item $\mEstParameterSet \coloneqq \{ \mEstParameter \in \mathbb R^p \mid \mEstConstrint(\mEstParameter)=0\}$ for some vector-valued constraint function $\mEstConstrint$ such that $\mEstParameterSet $ is bounded.
        \item The Jacobian matrix $\nabla \mEstConstrint( \mEstParameterEstN)$ has full rank for all $n$.
        \item $\En \left[ \nabla^2 \mEstLoss_{\mEstParameter}(v)\right]$ is invertible for all $\mEstParameter$.
        \item $\mEstParameterTrue$ is the unique minimizer of $ \E[\mEstLoss_\mEstParameter(\semVector)]$
    \end{itemize}

    Introduce  the definitions
    $\Jtrue \coloneqq \cov[\nabla \mEstLoss_{\mEstParameterTrue}(\semVector)]$,
    $\Ktrue \coloneqq  \E[ \nabla^2 \mEstLoss_{\mEstParameterTrue}(\semVector)]$ and
    $\PiTrue$ is an orthogonal projector in the complement of the range of the jacobian $\nabla \mEstConstrint( \mEstParameterTrue)$.
    Then we can establish the convergence
    \[\sqrt{n} ( \mEstParameterEstN - \mEstParameterTrue) \convd \normal(0,\Ktrue^{-1}\PiTrue\Jtrue\PiTrue\Ktrue^{-1}).\]
\end{theorem}

\begin{proof}

    Uniform weak law of large numbers holds, and $\mEstParameterSet$ must be compact since bounded and closed, so we have that $\mEstParameterTrue$ is consistently estimated by $\mEstParameterEstN$

    Let $\Qn$ be a matrix whose orthonormal columns spans the range of $\nablatheta \mEstConstrint( \mEstParameterEstN)$ (as in e.g. QR factorization).
    Construct an orthogonal matrix $[\Qn \, \Un]$.
    Now, $\Qn$ is a ON basis for the normal of the feasible set $\mEstParameterSet$, and $\Un$ is a ON basis for the tangent cone of $\mEstParameterSet$ as $\mEstParameterEstN$.

    Begin by a mean-value expansion of $\En \left[ \nablatheta \mEstLoss_{ \mEstParameterEstN}(\semVector)\right]$.
    \begin{align}
        \En[ \nablatheta \mEstLoss_{ \mEstParameterEstN}(\semVector)] & = \En[ \nablatheta \mEstLoss_{\mEstParameterTrue}(\semVector)] + \En[ \nablatheta^2  \mEstLoss_{\tilde \mEstParameter}(\semVector)] ( \mEstParameterEstN - \mEstParameterTrue)
    \end{align}
    We have that $I=[\Qn\,\Un]\begin{bmatrix}\Qn^\T\ \\ \Un^\T\end{bmatrix}$
    \begin{align}
         & [\Qn\,\Un]\begin{bmatrix}\Qn^\T \\ \Un^\T\end{bmatrix}\En [\nablatheta \mEstLoss_{ \mEstParameterEstN}(\semVector)]                                                                                                                  \\
         & = [\Qn\,\Un]\begin{bmatrix}\Qn^\T \\ \Un^\T\end{bmatrix}  \En [\nablatheta\mEstLoss_{\mEstParameterTrue}(\semVector)] +  \En [\nablatheta^2\mEstLoss_{\tilde \mEstParameter}(\semVector)] ( \mEstParameterEstN - \mEstParameterTrue)
    \end{align}
    By definition $\Un^\T\nablatheta g( \mEstParameterEstN)=0$, and from first order optimality conditions $\nablatheta \mEstLoss_{ \mEstParameterEstN}$ is in the range of $\nablatheta g( \mEstParameterEstN)$, so $\Un^\T\nablatheta \mEstLoss_{ \mEstParameterEstN} =0$.

    Rearranging, and using the assumption of invertibility of $\En [\nablatheta^2 \mEstLoss_{\tilde \mEstParameter}(\semVector)]$, we get
    \begin{align}
        \label{eq:mEst:before_any_limit}
         & ( \mEstParameterEstN - \mEstParameterTrue)=                                                                            \\
         & \En \left[\nablatheta^2 \mEstLoss_{\tilde \mEstParameter}(\semVector)\right]^{-1} [\Qn\,\Un]\begin{bmatrix}\Qn^\T \left(\En \left[ \nablatheta \mEstLoss_{ \mEstParameterEstN}(\semVector) -  \nablatheta \mEstLoss_{\mEstParameterTrue}(\semVector)\right] \right)  \\  -\Un^\T \En \left[ \nablatheta \mEstLoss_{\mEstParameterTrue}(\semVector)\right] \end{bmatrix}
    \end{align}
    Next, we will analyze a certain subexpression separately. Introduce $\PiTrue = \Utrue\Utrue^\T$ and $\PiN = \Un\Un^\T$.
    \begin{align}
         & \sqrt{n} \PiN \En \left[ \nablatheta \mEstLoss_{\mEstParameterTrue}(\semVector)\right] =                                                                                                                                                                                       \\
         & \quad \PiN \sqrt{n} \left( \En \left[ \nablatheta \mEstLoss_{\mEstParameterTrue}(\semVector)\right] - \E \left[ \nablatheta \mEstLoss_{\mEstParameterTrue}(\semVector)\right] \right) + \PiN \sqrt{n}  \E \left[ \nablatheta \mEstLoss_{\mEstParameterTrue}(\semVector)\right]
    \end{align}
    The first term converges to $\normal(0,\PiTrue\Jtrue\PiTrue)$ in distribution. The second term converges to zero in probability,  so
    \begin{align}
        \sqrt{n} \PiN \En \left[ \nablatheta \mEstLoss_{\mEstParameterTrue}(\semVector)\right] \convd \normal(0,\PiTrue\Jtrue\PiTrue)
    \end{align}
    Finally, we can take the limit of equation \eqref{eq:mEst:before_any_limit}.
    {\footnotesize\begin{multline}
        \sqrt{n} ( \mEstParameterEstN - \mEstParameterTrue) = \\
        \sqrt{n}\underbrace{[\En \left[ \nablatheta^2 \mEstLoss_{\tilde \mEstParameter}(\semVector)\right]^{-1}}_{\convp K^{-1}}  \underbrace{\Qn\Qn^\T}_{\convp \Qtrue\Qtrue^\T} \underbrace{\left(\En \left(\nablatheta \mEstLoss_{ \mEstParameterEstN}(\semVector)\right] - \En \left[\nablatheta \mEstLoss_{\mEstParameterTrue}(\semVector)\right] \right)}_{\convp 0} \\
        -  \underbrace{[\En \left[ \nablatheta^2 \mEstLoss_{\tilde \mEstParameter}(\semVector)\right]^{-1}}_{\convp \Ktrue^{-1}} \underbrace{ \sqrt{n} \PiN \En \left[ \nablatheta \mEstLoss_{\mEstParameterTrue}(\semVector)\right]}_{\convd \normal(0,\PiTrue\Jtrue\PiTrue)}
    \end{multline}
    }

    For all terms converging in probability we have been using the uniform weak law of large numbers, so we rely on compactness of $\mEstParameterSet$, and the suitable smoothness of the functions depending on $\semVector$.
    We need, for example, the continuity of matrix inversion, QR factorization and orthogonal complements. W
    e use Slutskys theorem to multiply the terms.

    Finally we see $\sqrt{n} ( \mEstParameterEstN - \mEstParameterTrue) \convd \normal(0,\Ktrue^{-1}\PiTrue\Jtrue\PiTrue\Ktrue^{-1})$
\end{proof}

\begin{corollary} \label{cor:constrained_m_est}
    The asymptotic distribution of Theorem~\ref{thm:constrained_m_estimation} can be reformulated by standardizing it, and plugging in estimates (e.g. $\Kn$) in the place of the population optimal expressions (e.g. $\Ktrue$).
    \[\sqrt{n}\mEstCovarianceN^{-1/2} ( \mEstParameterEstN - \mEstParameterTrue) \convd \normal(0,\eye).\]
    with the introduction of
    \[ \mEstCovarianceN \coloneqq \Kn^{-1}\PiN\Jn\PiN\Kn^{-1}  \]
    \[\Kn\coloneqq \En \left[ \nablatheta^2 \mEstLoss_{\mEstParameterEstN} (\semVector)\right]\]
    \[ \Jn \coloneqq \En[\nabla \mEstLoss_{\mEstParameterEstN}(v)\nabla \mEstLoss_{\mEstParameterEstN}(v)^\T]
        -\En[\nabla \mEstLoss_{\mEstParameterEstN}(v)]\En[\nabla \mEstLoss_{\mEstParameterEstN}(v)]^\T\]
\end{corollary}
\begin{proof}
    This follows from the consistency of plug-in-estimators \citep[Theorem 12.2]{wooldridge_econometric_2010}.
\end{proof}

%% file: sections/numerics.tex
\section{Numerical Illustrations}
\label{sec:numerics}

In the following experiments, data was generated using a linear \scm{} \eqref{eq:sem:linear_sem} with a matrix $\semCoeffMat$ that is either fixed or random. For random \DAG{}-matrices, we follow \citet[section 4.1]{yu_daggnn_2019}: Let $\dNodes$ be the number of nodes in a \scm{}. Let $k$ be the expected number of edges in a randomly generated \DAG{}. Let $M$ be a random strictly subtriangular matrix where entries are drawn $\operatorname{Bernoulli}(2k/(\dNodes-1))$. Let $P$ be a random permutation matrix. Let $C$ be uniformly drawn from the interval $[0.5,2]$, and set $\semCoeffMat=P^\T (C\hadamard M) P$.

The random vector $\semNoise$ in \eqref{eq:sem:linear_sem} has elements with unit variance and are drawn independently as either Normal(0,1), Exp(1) or Gumbel(0,$6/\pi^2$)). Data was also centered before any other processing.

Throughout all runs, the nominal miscoverage level was set to $\confidenceLevel=5\%$ and $\dagTolerance=10^{-7}$.

\begin{remark}
     In the supplementary material, we study deviations from the linear data model, in which case the average causal effect \eqref{eq:def:averageCausalEffectTarget} of the optimal linear model is still defined.
\end{remark}
\begin{remark}
     In all cases when the data generator is a linear \scm{} with Gaussian noise, we apply Isserlis' theorem to equation \eqref{eq:mest:score_variance}, $\En\left[ \semVector_i \semVector_q \semVector_o \semVector_k \right]-
          \En\left[ \semVector_i \semVector_q \right]\En\left[ \semVector_o \semVector_k \right] =
          \En\left[ \semVector_i \semVector_o \right]\En\left[ \semVector_q \semVector_k \right]
          +\En\left[ \semVector_i \semVector_k \right]\En\left[ \semVector_q \semVector_o \right]
     $. This reduction is especially helpful in high dimensions, when $\dNodes$ is large.
\end{remark}

\subsection{Numerical Search Method} \label{subsection:numerical_search}
In the examples below, we construct the confidence interval \eqref{eq:confidence_set_for_ace} by numerically solving problem \eqref{eq:def:thetan}. Here we use the augmented Lagrangian method \citep{nocedal_numerical_2006}, but other search methods are possible as well.

We define the augmented Lagrangian and the equality converted constraint as
\begin{align}
     \label{eq:def:augLag}
     \augLag(\mEstParameter,\augLagSlack,\augLagLagMul, \augLagPen) = \En \left[ \mEstLoss_\mEstParameter(\semVector)\right] + \augLagLagMul \augLagContraint(\mEstParameter, \augLagSlack) + \frac{\augLagPen}{2}\augLagContraint(\mEstParameter, \augLagSlack)^2
\end{align}
\[\augLagContraint(\mEstParameter, \augLagSlack) = \hFun(\matop(L\mEstParameter))+\augLagSlack^2-\dagTolerance \]

The method alternates between the minimization over primal variables ($\mEstParameter$,$\augLagSlack$) and maximization over dual variables ($\augLagLagMul$), starting from a few initialization points, as explicated in  Algorithm~\ref{algo:augLag}.

\begin{algorithm}[ht!]
     \DontPrintSemicolon
     \KwIn{$\mEstParameter^0$,$\augLagSlack^0$,$\augLagPen^0$, $\augLagLagMul^0$, $\augLagMinImprovement$, $\augLagPenMul$,$\augLag$,$\augLagConstraintTol$,$\augLagPenMax$,$\augLagContraint$}
     \KwOut{$\mEstParameterEstN$ }
     $k=0$\;
     \nl\While{
          $\augLagContraint(\mEstParameter^{\augLagIter}, \augLagSlack^{\augLagIter}) > \augLagConstraintTol$
          \textbf{ and }
          $\augLagPen < \augLagPenMax$
          \label{algo:augLag:stopCondition}
     }{
          \nl$\mEstParameter^{k+1},\augLagSlack^{k+1} = \argmin_{\mEstParameter,\augLagSlack} \augLag(\mEstParameter,\augLagSlack,\augLagLagMul^\augLagIter, \augLagPen^\augLagIter)$\; \label{algo:augLag:inner_problem}
          $\augLagLagMul^{\augLagIter+1} = \augLagLagMul^{\augLagIter} + \augLagPen^\augLagIter \augLagContraint(\mEstParameter^{\augLagIter+1}, \augLagSlack^{\augLagIter+1})$ \;
          \eIf{
               $\augLagContraint(\mEstParameter^{\augLagIter+1}, \augLagSlack^{\augLagIter+1}) >  \augLagMinImprovement \augLagContraint(\mEstParameter^{\augLagIter}, \augLagSlack^{\augLagIter})$
          }{
               $\augLagPen^{\augLagIter+1} =\augLagPenMul \augLagPen^{\augLagIter}$ \;
          }{
               $\augLagPen^{\augLagIter+1} = \augLagPen^{\augLagIter}$\;
          }
          $k=k+1$\;
     }
     \KwRet{$\mEstParameterEstN=\mEstParameter^{\augLagIter+1}$}\;
     \caption{Augmented Lagrangian Method}
     \label{algo:augLag}
\end{algorithm}

The minimization problem on line~\ref{algo:augLag:inner_problem} is solved via the L-BFGS-B-implementation in the python library \texttt{scipy.optimize}, which in turn utilizes the 3.0 version of the FORTRAN library of \citet{zhu_algorithm_1997}. Since this is a local minimizer, we use the previous optimal primal variables $\mEstParameter^{k},\augLagSlack^{k}$ as the starting point.

The parameters have default values set to $\mEstParameter^0=0$, $\augLagSlack^0=10$, $\augLagPen^0=1$, $\augLagLagMul^0=0$, $\augLagMinImprovement=1/4$, $\augLagPenMul=2$, $\augLagConstraintTol=10^{-12}$, $\augLagPenMax = 10^{20}$. Note that $\augLagConstraintTol$ must be significantly smaller than $\dagTolerance$, which in turn should be smaller than $\dagToleranceMax$. Thus it is advisable to verify that the choice of $\augLagConstraintTol$ is sufficiently small in a given problem. The threshold $\augLagPenMax$ is introduced for numerical stability.

The augmented Lagrangian method is guaranteed to find a local minimizer $\mEstParameterEstN$, under a certain set of assumptions \citep[Theorem 17.6]{nocedal_numerical_2006}. One of these is constraint qualification at the minimizer, in this case demanding $\nabla \augLagContraint(\mEstParameter_*,\augLagSlack_*) \neq 0$ at the optimal primal variables $\mEstParameter_*,\augLagSlack_*$. For $\dagTolerance=0$ this do not hold, but it does so for $\dagTolerance > 0$, see Lemma~\ref{lemma:nonconvexWset} in the supplementary material for a proof. Finding the minimum for $\dagTolerance \to 0$ will thus require $\augLagPen \to \infty$, and we have introduced the stop condition $\augLagPenMax$ on line~\ref{algo:augLag:stopCondition} for practical reasons.

To compute $\averageCausalEffectTarget$ we replace $\En[..]$ in \eqref{eq:def:augLag} with $\E[..]$, which has a closed-form expression.

\subsection{Baseline Comparison}
\label{subsection:correctly_identify_adjustment}
\label{subsection:same_as_ols}

We first compare the proposed confidence interval $\averageCausalEffectSet_{\nData,\confidenceLevel}$ in \eqref{eq:confidence_set_for_ace} with a standard \OLS-based confidence interval $\regCoefficientSet_{\nData,\confidenceLevel}$ for \eqref{eq:partial_regression_coeff} that is computed using HC0 standard errors \citep{wooldridge_econometric_2010}. To use \OLS{} we must specify a set of control variables, which we take to be $\adjustmentVar$. When this set is valid, we expect $\averageCausalEffectSet_{\nData,\confidenceLevel}$ and $\regCoefficientSet_{\nData,\confidenceLevel}$ to be similar. When the set is invalid, we expect them to diverge.

We use the linear Gaussian data model with the matrix in \eqref{eq:sem:linear_sem} set to be either
\[\semCoeffMat '=\begin{bmatrix}
          0 & 0 & 1 & 0 \\
          0 & 0 & 1 & 0 \\
          0 & 0 & 0 & 0 \\
          1 & 1 & 0 & 0 \\
     \end{bmatrix} \text{ or } \semCoeffMat '' = \begin{bmatrix} 0& 0.4& 0\\ 0& 0& 0 \\ 0.7 & 0.2& 0 \end{bmatrix} \]

The graph of $\semCoeffMat '$ is illustrated in Figure~\ref{fig:collider_dag}, while Figure~\ref{fig:collider_asymptotics} demonstrates the ability of $\averageCausalEffectSet_{\nData,\confidenceLevel}$ to correctly infer $\averageCausalEffectTarget$ without specifying a set of control variables. By contrast, $\regCoefficientSet_{\nData,\confidenceLevel}$ is clearly biased from incorrectly controlling for the collider $\adjustmentVar_1$.

Corresponding results for $\semCoeffMat ''$ are shown in Figure~\ref{fig:fork_asymptotics}. As expected, the resulting intervals $\averageCausalEffectSet_{\nData,\confidenceLevel}$ and $\regCoefficientSet_{\nData,\confidenceLevel}$ are virtually identical since $\adjustmentVar$ constitutes a valid set of control variables.
\begin{figure}[htbp]
     \centering
     \input{tikz/3node_fork_chart.tikz}
     \caption{$(1-\confidenceLevel)$-confidence intervals for $\averageCausalEffectTarget$ computed under a linear Gaussian \scm{} with matrix $\semCoeffMat ''$, for which $z$ is valid control variable.}
     \label{fig:fork_asymptotics}
\end{figure}

\subsection{Calibration and Normality}\label{subsection:calibration}
To assess the calibration of $\averageCausalEffectSet_{\confidenceLevel,\nData}$, we set $\nData$ to be $10^2$ or $10^4$ and generate repeated datasets from a linear Gaussian data model with matrix \[\semCoeffMat=\begin{bmatrix}
          0 & -2  & 1.6 & 0    \\
          0 & 0   & 0   & 0    \\
          0 & 1.2 & 0   & -0.5 \\
          0 & 0   & 0   & 0    \\
     \end{bmatrix}\]
corresponding to a graph illustrated in Figure~\ref{fig:calibration_dag}.
\begin{figure}
     \centering
     \input{tikz/calibration_dag.tikz}
     \caption{Causal structure of $\semCoeffMat$ in experiment for Calibration and Normality check, where $\adjustmentVar=[\adjustmentVar_1,\,\adjustmentVar_2]$ is  not a valid set of control variables.}
     \label{fig:calibration_dag}
\end{figure}
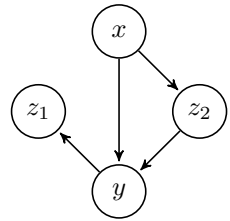

The coverage probability $\Prob ( \averageCausalEffectTarget \in  \averageCausalEffectSet_{\alpha, \nData}  )$ was estimated to be $94.6 \%$ and $94.9\%$ for $\nData=10^2$ and $10^4$, respectively, using $1000$ Monte Carlo simulations. This is close to $1-\alpha=95\%$ and corroborates Theorem~\ref{thm:confidence_set_for_ace}.
Figure~\ref{fig:calibration_qq} supports the result further by showing a Normality plot for the point estimate $\averageCausalEffectEstN$ over all Monte Carlo simulations.

\begin{figure}[htbp]
     \centering
     \input{tikz/calibration_qq.tikz}
     \caption{Normal probability plot for realizations of $\averageCausalEffectEstN$. Approximate normality is achieved even under moderate sample sizes.}
     \label{fig:calibration_qq}
\end{figure}
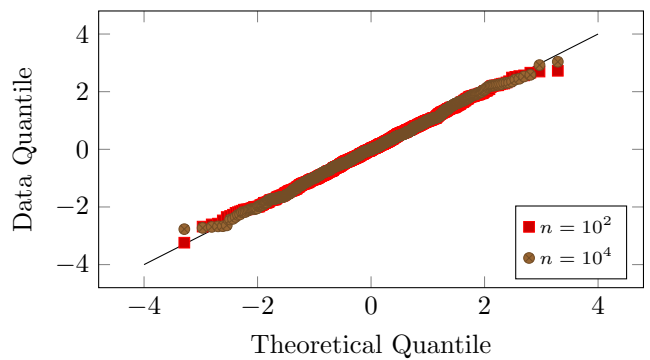

\subsection{Comparison With a Causal Discovery Method}

We compare our method with an alternative method of inferring the average causal effect by learning a linear \scm{} adjacency matrix $\semCoeffMat$ using DirectLiNGAM  \citep{shimizu_directlingam_2011, hyvarinen_pairwise_2013}. Then we can compute bootstrap confidence intervals, although they lack theoretical coverage guarantees. We used the official python implementation, version 1.5.1 from PyPI \url{https://pypi.org/project/lingam/1.5.1/}.

We generate a random adjacency matrix $\semCoeffMat$ for a graph on $\dNodes=10$ nodes and $k=1$, but with the random seed set to the lowest nonnegative integer that yielded a nonzero $\averageCausalEffect$ to make the comparison interesting. We use $\nData=10^{4}$ observations.

For LiNGAM, we computed the confidence interval (CI) using 100 bootstrap samples. For a comparable evaluation of its coverage, we considered the target quantity $\averageCausalEffectTarget$ to be the effect obtained when using LiNGAM with a large numbere of data points ($\nData'=10^6$). 100 Monte Carlo runs were used and the results are presented in Table~\ref{tab:lingam_compare}.

\begin{table}
     \centering
     \caption{Comparison of empirical coverage rate (CR) and the average width of the Confidence Interval (CI) for LiNGAM Bootstrap CI and the CI $\averageCausalEffectSet_{\confidenceLevel,\nData}$ proposed in this article. The nominal CR was set to exceed $1-\alpha = 95\%$}\label{tab:lingam_compare}
     \input{data/lingam_table}
\end{table}

The results show that when data is Gaussian, our proposed method yields both well-calibrated and tighter CIs, than LiNGAM method which has a very wide CI. This expected as LiNGAM was designed for non-Gaussian data. Indeed, for the non-Gaussian examples, LiNGAM produces tighter CIs but they all undercover. By constrast, our method produces more conservative CIs that do not undercover and yield consistent inferences.

\subsection{Sensitivity with Respect to \DAG{} tolerance}
\label{subsection:numerics:sensitivity}
Let $\averageCausalEffectTarget(\dagTolerance)$ denote the average causal effect  \eqref{eq:def:averageCausalEffectTarget} when setting a specific value $\dagTolerance$ in \eqref{eq:def:semcoefftrue}. When data-generating process is given by a linear \scm  \eqref{eq:sem:linear_sem}, we have that the approximation gap $|\averageCausalEffect - \averageCausalEffectTarget(0)| = 0$, where $\averageCausalEffect$ is given by \eqref{eq:def:averageCausalEffectInSem}.
The gap should decrease with $\dagTolerance$ such that ideally $\lim_{\dagTolerance \to 0} |\averageCausalEffect - \averageCausalEffectTarget(\dagTolerance)|=0$  and, moreover. An analytical study is, however, beyond the scope of the tools considered herein and we therefore resort to a numerical sensitivity study.

First, we generate random \DAG{}-matrices $\semCoeffMat$. For every $\semCoeffMat$, we form the numerically approximation $\hat{\averageCausalEffectTarget}(\dagTolerance)$ by replacing $\En$ with the closed for expression for $\E$ in \eqref{eq:def:augLag}. In Figure~\ref{fig:epsilon-limit}, we illustrate the approximation gap $|\averageCausalEffect - \hat{\averageCausalEffectTarget}(\dagTolerance)|$. As expected the gap decreases sharply with $\dagTolerance$, until we reach finite precision effects arising mainly from the L-BFGS-B implementation.

\begin{figure}
     \centering
     \input{tikz/dagtol_plot.tikz}
     \caption{The error between $\averageCausalEffect$ \eqref{eq:def:averageCausalEffectInSem} for a randomly generated matrix $\semCoeffMat$ and the numerically evaluated $\hat{\averageCausalEffectTarget}(\dagTolerance)$ from \eqref{eq:def:averageCausalEffectTarget} and \eqref{eq:def:semcoefftrue}, over a range of $\dagTolerance$. Each solid line corresponds to the error for a randomly drawn matrix, with a corresponding value of $\dagToleranceMax$ shown as a vertical grey dashed line.
          For $\dagTolerance \lesssim 10^{-7}$ the numerical precision of our numerical solver limits the precision of the results.
     }
     \label{fig:epsilon-limit}
\end{figure}

For some of the random matrices, we notice that when $\dagTolerance > \dagToleranceMax$ we obtain unreliable approximations. A more detailed discussion is provided in Section~\ref{subsection:extra_sensitivity} in the supplementary material.

In the work of \citet{ng_ontheconvergence_2020}, it is shown that the convergence guarantees for augmented Lagrangian method do not hold and that its precision is finite as it terminates when the quadratic penalty $\rho$ approaches infinity --- in agreement both with our theoretical and experimental results.

%% file: tikz/3node_fork_chart.tikz
\begin{tikzpicture}[baseline]
    \pgfplotstableread[col sep=comma]{./data/3node_fork_summary.csv}{\datatable};
    \begin{semilogxaxis}[
        width=\columnwidth,
        height=.6\columnwidth,
        xlabel={No. of data points, $\nData$},
        ylabel={Causal parameter, $\averageCausalEffect$},
        xmin=90,
        xmax=11000,
        legend style={font=\scriptsize, at={(0.5,1.10)},
		anchor=south,legend columns=-1},
    ]
        \addplot+ [ only marks, mark=*, mark size=1pt, error bars/.cd, y dir=both, y explicit] table [x=m_obs, y=ace_value, y error=q_ace_standard_error] {\datatable};
        \addplot+ [only marks, mark=*, mark size= 1 pt, error bars/.cd, y dir=both, y explicit] table [x=m_obs, y=ols_value, y error=q_ols_standard_error] {\datatable};
        \addplot [no markers] table [x=m_obs, y=ace_circ] {\datatable};
        \legend{{$\averageCausalEffectSet_{\confidenceLevel,\nData}$},{$\regCoefficientSet_{\confidenceLevel,\nData}$},{$\averageCausalEffectTarget$}}
    \end{semilogxaxis}
\end{tikzpicture}

%% file: tikz/calibration_dag.tikz
\begin{tikzpicture}[->,>=stealth',shorten >=1pt,auto,node distance=1.5cm,semithick]
  \node[circle,draw, minimum size=20pt] (Z2) {$\adjustmentVar_1$};
  \node[circle,draw, minimum size=20pt] (X) [above right of=Z2] {$\decisionVar$};
  \node[circle,draw, minimum size=20pt] (Y) [below right of=Z2] {$\outcomeVar$};
  \node[circle,draw, minimum size=20pt] (Z1) [below right of=X] {$\adjustmentVar_2$};

  \path (Z1) edge   (Y)
  (X) edge   (Z1)
  (X) edge   (Y)
  (Y) edge   (Z2)
            ;
\end{tikzpicture}

%% file: tikz/calibration_qq.tikz
\begin{tikzpicture}[baseline]
    \pgfplotstableread[col sep=comma]{./data/calibration_qq.csv}{\datatable};
    \begin{axis}[
        width=\columnwidth,
        height=.6\columnwidth,
        xlabel={Theoretical Quantile},
        ylabel={Data Quantile},
        legend pos=south east,
        legend style={font=\scriptsize},
    ]
        \addplot [no marks, domain=-4:4] {x};
        \addplot+ [only marks] table [x=ace_n_theoretical_z, y=0] {\datatable};
        \addplot+ [only marks] table [x=ace_n_theoretical_z, y=1] {\datatable};
        \legend{,{$\nData=10^2$},{$\nData=10^4$}}
    \end{axis}
\end{tikzpicture}

%% file: data/lingam_table.tex
\begin{tabular}{lllrr}
       \toprule
       Noise  & Method & CR    & Avg CI width & Avg $\averageCausalEffectEstN$ \\
       \midrule
       Normal & LiNGAM & 100\% & 2.01         & 0.64                           \\
              & our    & 99\%  & 0.15         & 1.79                           \\
       Exp    & LiNGAM & 92\%  & 0.08         & 1.77                           \\
              & our    & 100\% & 0.54         & 1.79                           \\
       Gumbel & LiNGAM & 85\%  & 0.07         & 1.77                           \\
              & our    & 100\% & 0.46         & 1.79                           \\
       \bottomrule
\end{tabular}

%% file: tikz/dagtol_plot.tikz
\begin{tikzpicture}[baseline]
    \pgfplotstableread[col sep=comma]{./data/dagtol_pgfplots.csv}{\datatable};
    \begin{axis}[
        width=\columnwidth,
        height=.6\columnwidth,
        xlabel={$\dagTolerance$},
        ylabel={$|{\averageCausalEffect - \hat{\averageCausalEffectTarget}(\dagTolerance)}|$},
        xmode=log,  
        ymode=log,
        legend style={font=\scriptsize},
        xmin=1e-9,
        xmax=1e2,
    ]
    \foreach \y in {
1.1392723396082403
,0.1937917407495311
,1.285846198734098
,1.1420203278694752
,0.2970456915441977
,1.5075613943267605
,1.5455653003752525
,0.9413888160321484
,0.6245695864066603
,1.750413866544636
        }
        \addplot +[mark=none,gray,dashed,thick] coordinates {(\y, 1e-8) (\y, 1e0)};
    \foreach \k in {0,1,...,9}
        \addplot [mark=none,color=black] table [x=dag_tolerance, y=ace_abs_err-\k] {\datatable};
    \end{axis}
\end{tikzpicture}

%% file: sections/conclusion.tex
\section{Conclusion}

We have developed a method that is capable of inferring average causal effects without the need to specify valid control variables, when the data-generating process can be described by a linear \scm. The methodology is based on characterizing \DAG{}-structures, which involve combinatorial constraints, using a continuously differentiable constraint. By considering a class of almost-\DAG{} matrices, we derive an asymptotically valid confidence interval building on a theory of equality-constrained M-estimation. The theoretical results were further corroborated in numerical studies with synthetic data.

Further research includes developing numerical search methods that are better tailored to approximate the constrained M-estimator upon which the confidence interval is based. Another research direction is the study of the properties of \eqref{eq:def:semcoefftrue} when $\dagTolerance \in ( 0 , \dagToleranceMax)$.

%% file: sections/suppl_mat.tex
\section{Supplementary material}
\subsection{Lemmas and proofs}

\begin{lemma} \label{lemma:m1_is_kronecker}
For the matrix $\semScaleMatrix = (\eye - \mutilatingMatrix \semCoeffMat^\T)^{-1}$, the element $\semScaleMatrix_{1i}$ is equal to the Kronecker delta $\kroneckerDelta_{1i}$, for $W\in \R^{\dNodes\times\dNodes}$ and $\mutilatingMatrix$ from equation \eqref{eq:def:mutilatingMatrix}.
\end{lemma}
\begin{proof}[Proof of Lemma \ref{lemma:m1_is_kronecker}]
Using Cramers rule $\semScaleMatrix_{1i}=\frac{1}{\det(\eye-\mutilatingMatrix\semCoeffMat^\T)}\cofactorMatrix_{i1}$, where $\cofactorMatrix$ is the cofactor matrix of $(\eye-\mutilatingMatrix\semCoeffMat^\T)$.

By definition of a cofactor as plus/minus a minor, and that the first row of $(\eye-\mutilatingMatrix\semCoeffMat^\T)$ is zero for all but the first element, $C_{i1}$ is zero for $i>1$, so $C_{i1} = \kroneckerDelta_{i1}\cofactorMatrix_{11}$

By Laplace expansion of $\det(\eye-\mutilatingMatrix\semCoeffMat^\T)$ along the first row
\[ \det(\eye-\mutilatingMatrix\semCoeffMat^\T) = \sum_{k=1}^d (\eye-\mutilatingMatrix\semCoeffMat^\T)_{1k} \cofactorMatrix_{1k} = \cofactorMatrix_{11} \]

We conclude $\semScaleMatrix_{1i} = \frac{1}{\cofactorMatrix_{11}}\kroneckerDelta_{i1}\cofactorMatrix_{11} = \kroneckerDelta_{1i}$
\end{proof}

\begin{proof}[Proof of lemma \ref{lemma:averageCausalEffectInSem}]
We need to show the result of equation \eqref{eq:def:averageCausalEffectInSem}.
Introduce $\semScaleMatrix = (\eye-\mutilatingMatrix \semCoeffMat)^{-1}$.

The proof follows by a direct computation, using Lemma \ref{lemma:m1_is_kronecker}. The noise covariance under the interventional distribution $\semInterventionNoiseCovariance$ is diagonal by assumption, which is also key.
\begin{align}
\averageCausalEffect(\semCoeffMat) 
&= \frac{\covint_{\semCoeffMat}[\decisionVar,  \outcomeVar]}{\varint_{\semCoeffMat}[\decisionVar]} \\
&= \frac{\covint_{\semCoeffMat}[\semVector,  \semVector]_{1,2}}{\covint_{\semCoeffMat}[\semVector,  \semVector]_{1,1}} \\
&= \frac{\sum_{i,j=1}^\dNodes M_{1j}M_{2i}\semInterventionNoiseCovariance_{ij}} {\sum_{i,j=1}^\dNodes M_{1j}M_{1i}\semInterventionNoiseCovariance_{ij}} \\ 
&= \frac{\sum_{i=1}^\dNodes M_{2i}\semInterventionNoiseCovariance_{i1}}{\semInterventionNoiseCovariance_{11}} \\ 
&= \frac{M_{21}\semInterventionNoiseCovariance_{11}}{\semInterventionNoiseCovariance_{11}} \\ 
&= M_{21}
\end{align}
This completes the proof.
\end{proof}

 We notice that there is nothing in the proofs of Lemma~\ref{lemma:m1_is_kronecker} and Lemma~\ref{lemma:averageCausalEffectInSem} specific about the first and second component - redefining the matrix $\mutilatingMatrix$ accordingly, it is straight forward to generalize the result if needed. 
 To keep the notation simple, we do stay with the convention that the first component is the one we intervene on, and that the second is the outcome of interest.

\begin{lemma}
\label{lemma:nabla_h}
	The function $h$ of \citet{zheng_dags_2018} has a closed form matrix gradient. It is ${\nabla h(W) = 2W \hadamard (\exp [ W\hadamard W ])^\T}$.
\end{lemma}
This formula is reported by \citet{zheng_dags_2018}, but without derivation. The result follows from liberal application of the chain rule.
\begin{proof}[Proof of Lemma \ref{lemma:nabla_h}]
$\frac{\partial}{\partial A_{i,j}} \tr A^k = k(A^{k-1})^\T_{i,j}$ by the product rule for derivation, and cyclicity of traces.

By series expansion and using the equation above
$\frac{\partial}{\partial A_{i,j}} \tr \exp [ A ] = (\exp[A])^\T_{i,j}$

We have that  $\frac{\partial (W \hadamard W)_{k,l}}{\partial W_{i,j} } = 2W_{i,j} \delta_{i,k} \delta_{j,l}$ using the Kronecker delta symbol.

The chain rule for differentiation now says
$\frac{\partial}{\partial W_{i,j}} \tr \exp [ W \hadamard W] 
= \sum_{k,l} \frac{\partial \tr \exp [ W \hadamard W]}{\partial (W\hadamard W)_{k,l}} \frac{\partial (W \hadamard W)_{k,l}}{\partial W_{i,j} }
= 2W_{i,j}\frac{\partial \tr \exp [ W \hadamard W]}{\partial (W\hadamard W)_{i,j}} 
= 2W_{i,j}(\exp[W \hadamard W])^\T_{i,j}$

The rest is a matter of notation and diffrentiating a constant.
\end{proof}

\begin{lemma} \label{lemma:nonconvexWset}
The set of all \DAG{}:s, $\semCoeffMatSet_0$ in \eqref{eq:def:semCoeffMatSet}, has the following properties
\begin{enumerate}
    \item  All points of $\semCoeffMatSet_0$ are boundary points (i.e., empty interior)
    \item $\semCoeffMatSet_0$ is a direct sum of linear subspaces, so it is a unbounded set, and a cone
    \item $\semCoeffMatSet_0$ is nonconvex. The convex hull of $\semCoeffMatSet_0$ is the set of all real $\dNodes\times\dNodes$-matrices.
    \item $\hFun(\semCoeffMat)=0$ iff $\nabla h(W) =0$.
\end{enumerate}
\end{lemma}

\begin{proof}[Proof of Lemma \ref{lemma:nonconvexWset}]
Only point four is a nontrivial result, as the others have a direct geometrical interpretation.

The first point follows from the fact that for $q$ being any matrix with a nonzero on the diagonal, $\hFun(\semCoeffMat+\varepsilon q) > 0 \quad \forall \varepsilon>0$, even when $\semCoeffMat \in \semCoeffMatSet$

The second point follows from the fact that $\hFun(\semCoeffMat)=0$ iff $\semCoeffMat$ is the weighted directed adjacency matrix of a DAG, and positive scaling that matrix will not affect the cyclicity structure.

The third point: Consider the example $w=\begin{bmatrix} 0 &1 \\ 0 & 0 \end{bmatrix}$. Then, $w,w^\T \in \semCoeffMatSet$, but $(w+w^\T)/2 \not \in \semCoeffMatSet_0$, so $\semCoeffMatSet_0$ is nonconvex. Consider also an arbitrary matrix  $W=\sum_{ij=1}^d w_{ij}\unitBasisMatrix^{ij}$. It is a convex combination of the matrices $\unitBasisMatrix^{ij}$, which all belong to $\semCoeffMatSet_0$. Since $W$ was arbitrary, all matrices are in the convex hull of $\semCoeffMatSet_0$.

The last point needs some more work, and is detailed below.

We start with the forward implication. Since any DAG $W$ is permutation similar to a strictly upper triangular matrix, $(\exp [ W\hadamard W ])^\T$ is permutation similar to a strictly lower triangular matrix, with the same similarity transformation.
$\nabla h(W)$ is therefore permutation similar to the elementwise product between a strictly upper and a strictly lower triangular matrix, which must be the zero matrix.

For the the reverse implication, assume $W$ is not a DAG, so it has some cycle of length $K$, and $1\leq K \leq d$.
Select $i$ and $j$ such that node $i$ and $j$ lies on that cycle. Now $W_{i,j} \neq 0$.
One can go from node $i$ to node $j$ in $1$ step, so one must be able to go from node $j$ to node $i$ in $K-1$ steps.
Therefore $ (W \hadamard W)^{K-1}_{j,i} \neq 0$. This makes sure that the exponential factor in $\nabla h(W)$ has a nonzero $i,j$-entry.

\[ \left[ (\exp [ W\hadamard W ])^\T \right]_{i,j} = \sum_{k=0}^\infty \frac{[(W \hadamard W)^k]_{j,i}}{k!} \neq 0 \]

\[\nabla h(W)_{i,j} = 2W_{i,j} \left[(\exp [ W\hadamard W ])^\T \right]_{i,j}\]

Since this is a product of two positive real numbers, we can conclude that $\nabla h(W)  \neq 0$.
\end{proof}

This result supplements the discussion of \citet[p.7]{zheng_dags_2018}. Not only is the \DAG{}:s the global minima of $\hFun$, but they are also the zeroes of $\nabla \hFun$.

The fourth point in Lemma~\ref{lemma:nonconvexWset} has during the time of writing this being reported in \citet[lemma 4]{wei_dags_2020}, but with a more different derivation technique valid for a slightly broader class of $h$-functions. It has also been reported in \citet[proposition 1]{ng_graph_2019}, with a proof technique very similar to ours.

\begin{lemma}
\label{lemma:lossAsQuadratic}
The least-squares objective, and its derivatives are
\begin{align}
\mEstLoss_\mEstParameter(v) = \frac{1}{2} (\mEstParametrization\mEstParameter- \vecop(\eye))^\T \left[\semNoiseCovariance^{-1} \kronecker vv^\T\right] (\mEstParametrization\mEstParameter- \vecop(\eye))
\end{align} 
and its gradient and hessian is
\begin{equation} \label{eq:mest:gradient} \nabla \mEstLoss_\mEstParameter(v) =  \mEstParametrization^\T \left[\semNoiseCovariance^{-1} \kronecker vv^\T\right] ( \mEstParametrization\mEstParameter- \vecop(\eye) ) \end{equation}
\[ \nabla^2 \mEstLoss_\mEstParameter(v) =  \mEstParametrization^\T \left[\semNoiseCovariance^{-1} \kronecker vv^\T\right]  \mEstParametrization \]
\end{lemma}
The proof is direct computation, after using the formula $\tr(A^{\T} Y^{\T} BX) = (\vecop(Y))^{\T}[A\kronecker B] \vecop(B)$.
\begin{proof}[Proof of Lemma \ref{lemma:lossAsQuadratic}]

Use the vec-trick $\tr(A^{\T} Y^{\T} BX) = \vecop(Y)^{\T}[A\kronecker B] \vecop(B)$, and find the objective.
\begin{align}
    \mEstLoss_\mEstParameter(v) &= \frac{1}{2}  \norm{\semNoiseCovariance^{-1/2}\left(\eye-\matop(\mEstParametrization\mEstParameter)^\T\right)v}^2 \\
    &= \frac{1}{2} \tr \left[ \semNoiseCovariance^{-1}\left(\matop(\mEstParametrization\mEstParameter)-\eye\right)^{\T}vv^\T \left(\matop(\mEstParametrization\mEstParameter)-\eye)\right) \right] \\
    &= \frac{1}{2} (\mEstParametrization\mEstParameter- \vecop(\eye))^\T \left[\semNoiseCovariance^{-1} \kronecker vv^\T\right] (\mEstParametrization\mEstParameter- \vecop(\eye))
\end{align} 

The rest is differentiation of a quadratic.
\end{proof}

\begin{lemma}\label{lemma:mest:symbols}
The quantities of Lemma~\ref{lemma:asymptotic_normal_mestparam} can be computed to be
\[\Kn = \mEstParametrization^\T   \left[ \semNoiseCovariance^{-1} \kronecker \En \left[  \semVector \semVector^\T \right] \right]\mEstParametrization\]
\[\PiN = \eye - (qq^\T)/(q^{\T}q)\] 
\[q = \mEstParametrization^\T\vecop(2\semCoeffEstN \hadamard (\exp [ \semCoeffEstN\hadamard \semCoeffEstN ])^\T)\]
\[\Jn = \mEstParametrization^\T \tilde \Jn \mEstParametrization \]
\begin{multline}
     (\tilde \Jn )_{d(j-1)+i,d(l-1)+k} =
    \sum_{q,r,o,p=1}^{\dNodes} \Big\{
    \big( \En\left[ \semVector_i \semVector_q \semVector_o \semVector_k \right]- \\
    \En\left[ \semVector_i \semVector_q \right]\En\left[ \semVector_o \semVector_k \right]\big)
    \semNoiseCovariance^{-1}_{j,r}
    \semNoiseCovariance^{-1}_{p,l}
    (\semCoeffMat-\eye)_{q,r}
    (\semCoeffMat-\eye)_{o,p}
    \Big\}
\end{multline}
\end{lemma}
\begin{proof}[Proof of Lemma~\ref{lemma:mest:symbols}]

The expression for $\Kn$ follows from Lemma~\ref{lemma:lossAsQuadratic}.
\begin{multline}
\Kn=
\En[\nabla^2 \mEstLoss_\mEstParameter(v)] = \\
\En \left[ \mEstParametrization^\T \left[\semNoiseCovariance^{-1} \kronecker vv^\T\right]  \mEstParametrization \right] 
=  \mEstParametrization^\T \left[\semNoiseCovariance^{-1} \kronecker \En \left[ vv^\T\right] \right]  \mEstParametrization 
\end{multline}

$\PiN$  is a projection matrix with respect to the orthogonal complement of $q\coloneqq \nabla_{\mEstParameter} \hFun(\matop(L\mEstParameterEstN))$. Since $q$ is a vector, projection on the orthogonal complement is $\PiN = \eye - (qq^\T)/(q^{\T}q)$. The expression $q=\mEstParametrization^\T\vecop(2\semCoeffEstN \hadamard (\exp [ \semCoeffEstN\hadamard \semCoeffEstN ])^\T)$ follows from Lemma~\ref{lemma:nabla_h}, and $\semCoeffEstN = \vecop{\mEstParametrization\mEstParameterEstN}$.

The derivation of $\Jn$ is an mostly tracking indices. Start with $\Jn = \En[\nabla \mEstLoss_{\mEstParameterEstN}(v)\nabla \mEstLoss_{\mEstParameterEstN}(v)^\T]
 -\En[\nabla \mEstLoss_{\mEstParameterEstN}(v)]\En[\nabla \mEstLoss_{\mEstParameterEstN}(v)]^\T$ and apply to  Lemma~\ref{lemma:lossAsQuadratic}. First factor out the $\mEstParametrization$ matrix of \eqref{eq:mest:gradient}, and then covert the rest into indices. Apply the index conversion for vectorizations $\vecop{A}_{d(j-1)+i}=A_{i,j}$ and for kronecker products $[A\kronecker{}B]_{d(i-1)+j,d(k-1)+l}=A_{i,k}B_{j,l}$ when $A$ and $B$ are $\dNodes\times\dNodes$ sized.
\end{proof}

The next lemma collects the assumption verification for applying Corollary~\ref{cor:constrained_m_est} in proof of Lemma~\ref{lemma:asymptotic_normal_mestparam}. Herein we use the redundant norm-constraint, that is in some parts skipped.
\begin{lemma}
\label{lemma:mest:technicalities}
Using the loss function \eqref{eq:def:mest:loss}, and the parameter set $\mEstParameterSet:=\{\mEstParameter \mid| \hFun\left(\matop(\mEstParametrization\mEstParameter\right)-\dagTolerance = 0 \land \norm{\theta}\leq B \}$, we see that
\begin{enumerate}
    \item The techincal conditions for M-estimation \citep[Theorem 12.2]{wooldridge_econometric_2010} holds.
    \item The loss function $\mEstLoss_\mEstParameter(\semVector)$ is two times continously diffrentiable in $\semVector$.
    \item $\mEstParameterSet \coloneqq \{ \mEstParameter \in \mathbb R^p \mid \mEstConstrint(\mEstParameter)=0\}$ for some vector-valued constraint function $\mEstConstrint$ such that $\mEstParameterSet $ is bounded.
    \item The Jacobian matrix $\nabla \mEstConstrint( \mEstParameterEstN)$ has full rank for all $n$.
    \item $\En \left[ \nabla^2 \mEstLoss_{\mEstParameter}(v)\right]$ is invertible for all $\mEstParameter$.
    \item $\mEstParameterTrue$ is the unique minimizer of $ \E[\mEstLoss_\mEstParameter(\semVector)]$
\end{enumerate}
\end{lemma}
\begin{proof}
First notice that \eqref{eq:def:mest:loss} is quadratic in $\mEstParameter$, but also in $\semVector$, which is more clearly seen in \eqref{eq:def:Wn}.
\begin{enumerate}
    \item The technical conditions are (a) that $\mEstParameterSet$ is compact, which follows from closed and boundedness (b) that $\mEstLoss_\mEstParameter(\semVector)$ is borel measurable in $\semVector$ for each $\mEstParameter$, which follow from being quadratic, (c) that $\mEstLoss_\mEstParameter(\semVector)$ is continuous in $\mEstParameter$ for each $\semVector$, which follows from being a quadratic and (d) that there is a dominating function $d(\semVector)\geq |\mEstLoss_\mEstParameter(\semVector)|$ for all $\mEstParameter$ so that $\E[d(\semVector)] < \infty$, which needs a few steps to prove. Observe
    \begin{align}
    |\mEstLoss_\mEstParameter(\semVector)|&=\frac{1}{2}\norm{\semNoiseCovariance^{-1/2}(\eye-\matop(L\theta))\semVector}_2^2\\
    &\leq\frac{1}{2}\sigma_1(\semNoiseCovariance^{-1/2})^2\sigma_1(\eye-\matop(L\theta))^2\norm{v}^2\\
    & \leq C\norm{v}^2\eqqcolon d(v),\end{align}
    where $\sigma_1$ denotes the largest singular value and \[C:=\frac{1}{2}\sigma_1(\semNoiseCovariance^{-1/2})^2\max_{\mEstParameter \in \mEstParameterSet} \sigma_1(\eye-\matop(L\theta))^2 ,\]
    utilizing compactness of $\mEstParameterSet$. Finally $\E[d(v)] = C\E[\norm{v}^2] = C\tr{}[(\eye-W^{\T})^{-1}\semNoiseCovariance(\eye-W)^{-1}] \leq \infty$, using the assumed data generating process \eqref{eq:sem:linear_sem}.
    \item $\mEstLoss_\mEstParameter(\semVector)$ is two times continously diffrentiable in $\semVector$, since it is a quadratic in $\semVector$
    \item The form of $\mEstParameterSet:=\{\mEstParameter \mid| \hFun\left(\matop(\mEstParametrization\mEstParameter\right)-\dagTolerance = 0 \land \norm{\theta}\leq B \}$ can be transformed into equality form by introduction of a slack variable $s$, so that $\mEstParameterSet:=\{\mEstParameter,s \mid| \hFun\left(\matop(\mEstParametrization\mEstParameter\right)-\dagTolerance = 0 \land \norm{\theta}+s^2-B=0 \}$, so $g(s,\mEstParameter)=\begin{bmatrix}\hFun\left(\matop(\mEstParametrization\mEstParameter\right)-\dagTolerance \\ \norm{\theta}+s^2-B\end{bmatrix}$.
    \item By lemma~\ref{lemma:nonconvexWset}, $\nabla \mEstConstrint( \mEstParameterEstN)$ is nonzero over $\mEstParameterSet$, but the gradient with respect to the slack is zero. Furthermore $\nabla_s[\norm{\theta}+s^2-B]=2s$, which is zero only for $s=0$, but we know from \ref{lemma:unconstrained_minimization} that $s\neq 0$. So the two components of $g$ must have linerarly independent gradients, and the jacobian has full rank. Do note that the slack-formulation used here is supressed from the formalism in the rest of the article, since it is an inactive constraint, making the proofs and text less clear with no gain.
    \item  $\En \left[ \nabla^2 \mEstLoss_{\mEstParameter}(v)\right] = \mEstParametrization^\T \left[\semNoiseCovariance^{-1} \kronecker \En[vv^\T]\right]  \mEstParametrization $, which almost surely has full rank. We ignore the measure zero case.
    \item The unicity of $\mEstParameterTrue$ we have to take by assumption, as discussed elsewhere in this article.
\end{enumerate}
\end{proof}

\begin{lemma}
\label{lemma:gradient_of_ace}
The gradient of the causal effect $\averageCausalEffect$ with respect to the parameter $\mEstParameter$ is 
\begin{align}
\left[ \nabla _\mEstParameter \averageCausalEffect(\mEstParameter) \right]_k = -\left( \left[\semScaleMatrix \mutilatingMatrix \kronecker \eye \right] \mEstParametrization \right)_{d+1,k} 
\end{align}
\end{lemma}
\begin{proof}[Proof of Lemma \ref{lemma:gradient_of_ace}]
Start from Lemma~\ref{lemma:averageCausalEffectInSem}. Apply derivation rules for matrix inverses, and utilize the unit basis matrices $\unitBasisMatrix^{i,j}$ which zero in every entry, except the $i,j$-entry.
\begin{align}
    \frac{\partial(\averageCausalEffect(W))}{\partial \semCoeffMat_{i,j}} &= \frac{\partial(\semScaleMatrix_{21})}{\partial \semCoeffMat_{i,j}}  \\ &=\sum_{k,l=1}^\dNodes \semScaleMatrix_{2k} \frac{\partial( \eye - \mutilatingMatrix\semCoeffMat^\T))_{kl}}{\partial \semCoeffMat_{i,j}} \semScaleMatrix_{l1} \\
    &=-\sum_{k,l=1}^\dNodes \semScaleMatrix_{2k} \mutilatingMatrix_{km} \unitBasisMatrix^{ij}_{lm} \semScaleMatrix_{l1}  \\
    &= -(\semScaleMatrix \mutilatingMatrix)_{2j}\semScaleMatrix_{i1} \\
    &= -\left[ \semScaleMatrix \mutilatingMatrix \kronecker \semScaleMatrix^\T  \right]_{d+1,\dNodes(j-1)+i} \\
\end{align}

As an aside, we can note that the matrix with these entries has a compact definition, $-(\left[ \semScaleMatrix \mutilatingMatrix \kronecker \semScaleMatrix^\T  \right]) = \frac{\partial \vecop(M^{T})}{\partial \vecop W}$. Armed with this expression and 
\begin{align}
    \frac{\partial \semCoeffMat_{i,j}}{\partial \mEstParameter_k} = \mEstParametrization_{\dNodes(j-1)+i,k}
\end{align}
we can compute 
\begin{align}
\left[ \nabla _\mEstParameter \averageCausalEffect(\mEstParameter) \right]_k 
&= \sum_{i,j=1}^\dNodes \frac{\partial(\averageCausalEffect(W))}{\partial \semCoeffMat_{i,j}}\frac{\partial \semCoeffMat_{i,j}}{\partial \mEstParameter_k} \\ 
=&-\left( \left[\semScaleMatrix \mutilatingMatrix \kronecker \eye \right]  \mEstParametrization \right)_{d+1,k}
\end{align}
\end{proof}

\subsection{Numerical Experiments}
\subsubsection{Detailed sensitivity study}\label{subsection:extra_sensitivity}
In section~\ref{subsection:numerics:sensitivity} we studied the impact of $\dagTolerance$ in relation to our causal effect measure $\averageCausalEffectTarget$. In this section, we provide additional results (in Figure~\ref{fig:sensitivity_details}) that shed more light on the behavior of the solution. 

The computations are performed as in in section~\ref{subsection:numerics:sensitivity}, but with 20 random graphs instead of 10, and a wider range of $\dagTolerance$

Comparing Figures \ref{fig:sensitivity_details}d and \ref{fig:sensitivity_details}b, we note that while setting $\dagTolerance > \dagToleranceMax$ yields an inaccurate non-\DAG{} matrix $\semCoeffOpt(\dagTolerance)$, it may occasionally produce accurate $\averageCausalEffectNumeric(\dagTolerance)$ depending on the unknown data-generating process and the nonlinear mapping in \eqref{eq:def:averageCausalEffectInSem}.

\begin{figure*}[ht!]
    \centering
    \input{tikz/dagtol_details.tex}
    \caption{Detailed graphs for the extended sensitivity analysis. We conclude that $\dagTolerance \rightarrow 0$ is a strong indication that $\semCoeffOpt(\dagTolerance) \rightarrow \semCoeffOpt(0)$..}\label{fig:sensitivity_details}
\end{figure*}

In Figure~\ref{fig:sensitivity_details:hfun} we see that to improve the \DAG{}-fidelity (quantified by $\hFun(\semCoeffMat)$), we need to reduce $\augLagConstraintTol$. However, in the numerical runs, we could see that required raising $\augLagPenMax$ further, which may lead to numerical inaccuracies.

\subsubsection{Linearity assumptions violations}
\label{subsection:linearity_assumption_violation}

All numerical experiments above been performed using data drawn from \emph{linear} \scm{}s. We now consider the behavior of the method when the data-generating process is non-linear and study the coverage of the target quantity $\averageCausalEffectTarget$. It is still defined in \eqref{eq:def:averageCausalEffectTarget} as the average causal effect of the optimal linear \scm{} (although it will diverge from the unknown distribution parameter $\averageCausalEffect$ depending on the type of nonlinearity). 

We use the same models as \citet{yu_daggnn_2019}:
\begin{enumerate}
    \item Linear: $\semVector=\semCoeffMat^\T\semVector+\semNoise$ where 
    \item Nonlinear 1: $\semVector=\semCoeffMat^\T \cos(\semVector+\vecOne) +\semNoise$,
    \item Nonlinear 2: $\semVector=2\sin (\semCoeffMat^\T(\semVector+0.5\cdot \vecOne)) + \semCoeffMat^\T(\semVector+0.5\cdot \vecOne) +\semNoise$
\end{enumerate}
The coefficient matrix $\semCoeffMat$ is generated as in section \ref{sec:numerics} and the random elements of $\semNoise$ are drawn independently as $\normal (0,1)$. Let $\vecOne$ denote a vector of ones, and $\cos(\cdot)$ and $\sin(\cdot)$ on vectors be defined entry-wise.
For each of these models $\nData=10^{3}$ data points are generated. 

We performed $200$ Monte Carlo runs and report the empirical coverage rate $CR$  of $\averageCausalEffectSet_{\alpha, \nData}$ in Table \ref{tab:nonlinear_results}, $\dNodes$ is the number of nodes in the \scm{} and $k$ denotes the number of number of expected edges per node. We find that in all cases the empirical coverage rate exceeds the target $1-\alpha = 95\%$, in accordance with the theory, but the confidence interval is more conservative in the nonlinear cases than the linear case. 

\begin{table}
    \centering
    \caption{Empirical coverage rates of $\averageCausalEffectSet_{\nData,\confidenceLevel\%}$ from numerical experiment on linear assumption violation. Nominal coverage set to $1-\confidenceLevel=95\%$.}\label{tab:nonlinear_results}
    \begin{tabular}{ccccc}
      \toprule
      $\dNodes$ & $k$ & linear&nonlinear1&nonlinear2\\
      \midrule
      5 & 1 & 98.0\% & 97.0\%& 99.5\% \\
      5 & 2 & 97.5\% & 96.5\%& 100.0\% \\
      10 & 1 & 96.0\% & 98.5\%& 99.5\% \\
      10 & 2 & 95.5\% & 96.5\%& 100.0\% \\
      \bottomrule
    \end{tabular}
\end{table}

\subsubsection{Misspecified latent covariance structure}
One of the major challenges of the method is the assumption of an approximately known latent covariance $\semNoiseCovariance$. This section explores the sensitivity to misspecification in this parameter.

First, we restate \citet[Theorem 9]{loh_high-dimensional_2014}. 
Let $W_1 \gg W_0$ if the directed graph encoded by $W_1$ is a supergraph of $W_0$. \emph{I.e.} for all indices $i,j$, $[W_0]_{i,j} \neq 0$ implies $[W_1]_{i,j} \neq 0$. The converse, $W_1 \not\gg W_0$ means that there is some component of $W_1$ that is zero, even though the corresponding  component of $W_0$ is not. 
Define the \emph{additive gap} $\xi $ to be the difference in expected squared loss between the optimal DAG adjacency matrix and the second best one among the non-supergraph-models. Compare the following with \eqref{eq:def:semcoefftrue}. Define
\begin{align}
\mathtt{score}(W)\coloneqq& \E\Big[ \norm{\semNoiseCovariance^{-1/2}  \left(\eye - \semCoeffMat^\T\right)\semVector }^2 \Big] \\
W_{0} \coloneqq& \argmin_{\semCoeffMat \in \semCoeffMatSet_0} \mathtt{score}(W) \\
\xi \coloneqq& \min_{\substack{\semCoeffMat \in \semCoeffMatSet_0\\ \semCoeffMat\not\gg \semCoeffMat_0 }} \left\{ \mathtt{score}(W)\right\} - \mathtt{score}(W_0) 
\end{align}
This gap is defined from the data generating process uniquely, and can only be computed if the the data generating latent covariance $ \semNoiseCovariance$ is known - at least up to a scale factor. 
When this is not known, we assume some latent variance structure $\widehat \semNoiseCovariance$, and quantify our misspecification by the condition number $\misspecCond$.

\begin{lemma}[Loh Bühlmann, Lemma 9]
If
\[\misspecCond \leq 1+\frac{\xi}{\dNodes} \label{eq:good noise covariance}\]
then $W_0 \in\argmin_{\semCoeffMat \in \semCoeffMatSet_0} \E\Big[ \norm{\assStruct^{-1/2}  \left(\eye - \semCoeffMat^\T\right)\semVector }^2 \Big]$. If the inqeuality is strict, then $W_0$ is the unique minimizer.
\end{lemma}

If the structure is correctly assumed, \emph{i.e.} $\semNoiseCovariance = \semNoiseScale \assStruct$ for some scaling factor $\semNoiseScale$, then 
\[ \min_{\semCoeffMat \in \semCoeffMatSet_0} \E\Big[ \norm{\assStruct^{-1/2}  \left(\eye - \semCoeffMat^\T\right)\semVector }^2 \Big] = s\dNodes\]
so we can estimate the scale factor $s$ from data, assuming that we have the correct latent covariance structure $\widehat \semNoiseCovariance$.\citep[Corollary 8]{loh_high-dimensional_2014} Denote this empirical estimate $\hat s$.

How does these results affect the confidence interval of Theorem~\ref{thm:cofidence_interval_for_ace}? We replace $\semNoiseCovariance $ in \eqref{eq:mest:score_variance} with $\hat \semNoiseScale \assStruct$ using the biased estimate of the scale $\semNoiseScale$.
\footnote{The estimate is most likely biased since most likely $\assStruct$ is not proportional to the true data generating $\semNoiseCovariance$.}
We conducted numerical studies aiming to illustrate that the confidence interval is good when $\misspecCond$ is small enough. 

We generate data as in \ref{subsection:calibration}, but with a random latent noise matrix $\semNoiseCovariance$. The matrix is diagonal, with entries drawn uniformly iid from from the interval $[1-\Delta,1+\Delta]$, and $\Delta=\frac{1-\kappa_{max}}{1+\kappa_{max}}$. We use $\widehat\semNoiseCovariance=\eye$ as before. This guarantees that $\misspecCond \leq \kappa_{max}$.

For each draw of $\nData$ data points, compute $\misspecCond$, as well as $\averageCausalEffectTarget$ and $\averageCausalEffectSet$ as described in section \ref{sec:numerics}.

\begin{figure}
    \centering
    \input{tikz/misspec_scatter.tikz}
    \caption{The average causal effect $\averageCausalEffectTarget$ is in general close to the true value, except when the condition number $\misspecCond$ becomes larger than some threshold value. This computation is not dependant on the number of data points drawn. Every run is marked with an $x$, and the true average causal effect is denoted with a dashed hosrizontal line, mostly occluded by the $x$-marks.}
\end{figure}
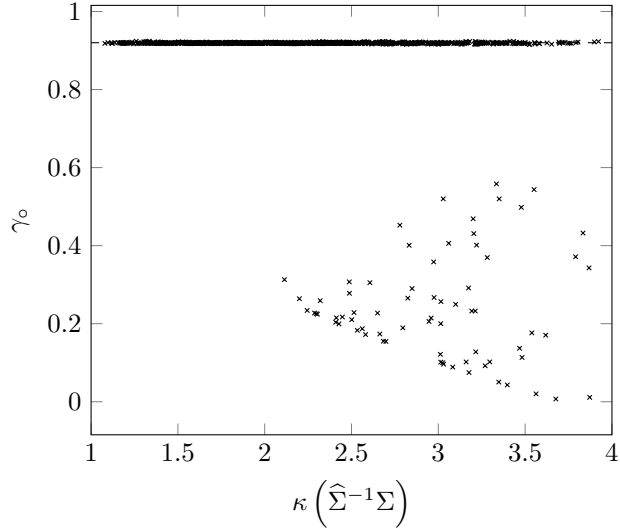

\begin{figure}
    \centering
    \input{tikz/misspec_coverage_100.tikz}
    \caption{For $\nData=100$. Empirical coverage, as the misspecification is increased. 1000 runs with random noise matrices $\semNoiseCovariance$ run. For each run, we have computed if $\averageCausalEffectTarget \in \averageCausalEffectSet$ or not. The runs have been binned in groups of $n_b=100$, and each bin $b$ has an empirical coverage rate $\hat p_b$ computed. The shaded area represent $\hat p \pm 2\sqrt{\frac{\hat p(1-\hat p)}{n_b}}$. In general, the misspecification voids the guarantee for the coverage rate, but as long as the misspecification is small, the coverage rate is close to the promised one. } \label{fig:misspec coverage 100}.
\end{figure}

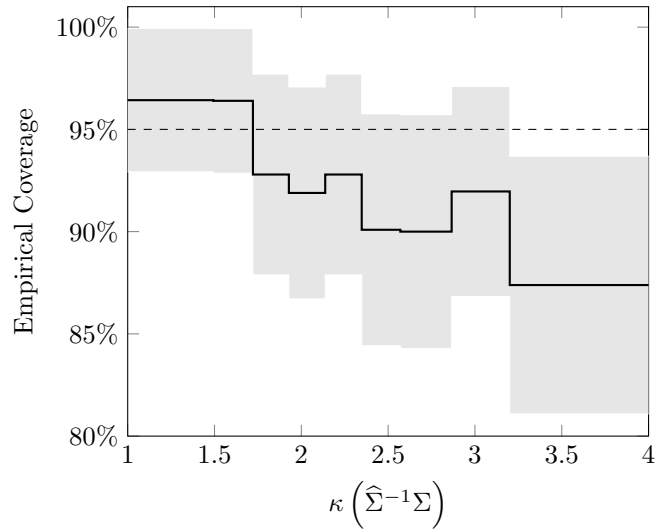
\begin{figure}
    \centering
    \input{tikz/misspec_coverage_10000.tikz}
    \caption{Setup as in Figure~\ref{fig:misspec coverage 100}, but $\nData=10000$.}
\end{figure}

%% file: tikz/dagtol_details.tex
\newcommand{\xmin}{1e-15}
\newcommand{\xmax}{1e2}

\begin{subfigure}[b]{0.4\textwidth}
\begin{tikzpicture}[baseline]
    \pgfplotstableread[col sep=comma]{./data/dagtol_pgfplots.csv}{\datatable};
    \begin{loglogaxis}[
        width=\textwidth,
        height=0.7\textwidth,
        xlabel={$\dagTolerance$},
        ylabel={$|\hat{\averageCausalEffectTarget}(\dagTolerance)|$},
        xmin=\xmin,
        xmax=\xmax
    ]
    \foreach \k in {0,1,...,19}
        \addplot [mark=none,color=black] table [x=dag_tolerance, y=ace_abs-\k] {\datatable};
    \end{loglogaxis}
\end{tikzpicture}
\subcaption{The average causal effect estimated for various $\dagTolerance$. Absolute value imposed to allow log-log-plot.}
\end{subfigure}
\hspace{0.1\textwidth}
\begin{subfigure}[b]{0.4\textwidth}
\begin{tikzpicture}[baseline]
    \pgfplotstableread[col sep=comma]{./data/dagtol_pgfplots.csv}{\datatable};
    \begin{loglogaxis}[
        width=\textwidth,
        height=0.7\textwidth,
        xlabel={$\dagTolerance$},
        ylabel={$|\averageCausalEffect - \averageCausalEffectNumeric(\dagTolerance)|$},
        xmin=\xmin,
        xmax=\xmax
    ]
    \foreach \k in {0,1,...,19}
        \addplot [mark=none,color=black] table [x=dag_tolerance, y=ace_abs_err-\k] {\datatable};
    \end{loglogaxis}
\end{tikzpicture}
\subcaption{The absolute error in the estimate of the causal effect. Smilar to figure \ref{fig:epsilon-limit}.}
\end{subfigure}

\vspace{.1\textwidth}
\begin{subfigure}[b]{0.4\textwidth}
\begin{tikzpicture}[baseline]
    \pgfplotstableread[col sep=comma]{./data/dagtol_pgfplots.csv}{\datatable};
    \begin{loglogaxis}[
        width=\textwidth,
        height=0.7\textwidth,
        xlabel={$\dagTolerance$},
        ylabel={$\hFun(\hat{\semCoeffMat}(\dagTolerance))$},
        xmin=\xmin,
        xmax=\xmax
    ]
    \foreach \k in {0,1,...,19}
        \addplot [mark=none,color=black] table [x=dag_tolerance, y=h_notears-\k] {\datatable};
    \end{loglogaxis}
\end{tikzpicture}
\subcaption{The constraint function $\hFun$ at the numerical approximation of the $\dagTolerance$-almost \DAG{} $\semCoeffOpt$. If the numerical solver is good and $\dagTolerance \leq \dagToleranceMax$, we should have $\hFun(\hat{\semCoeffMat(\dagTolerance}) \approx \dagTolerance$, which is what we observe down to circa $10^{-12}=\augLagConstraintTol$, the tolerated constraint violation of Algorithm~\ref{algo:augLag}. We can also see that when $\dagTolerance > \dagToleranceMax$, the solution does not depend on $\dagTolerance$.}
\label{fig:sensitivity_details:hfun}
\end{subfigure} 
\hspace{0.1\textwidth}
\begin{subfigure}[b]{0.4\textwidth}
\begin{tikzpicture}[baseline]
    \pgfplotstableread[col sep=comma]{./data/dagtol_pgfplots.csv}{\datatable};
    \begin{loglogaxis}[
        width=\textwidth,
        height=0.7\textwidth,
        xlabel={$\dagTolerance$},
        ylabel={$\norm{\semCoeffMat -  \hat{\semCoeffMat}(\dagTolerance) }_{\infty}$},
        xmin=\xmin,
        xmax=\xmax
    ]
    \foreach \k in {0,1,...,19}
        \addplot [mark=none,color=black] table [x=dag_tolerance, y=max-metric-\k] {\datatable};
    \end{loglogaxis}
\end{tikzpicture}
\subcaption{The maximum error in the point estimate of the adjacency matrix $\semCoeffMat$. The results indicate $\dagTolerance \to 0$ is a necessary condition to retrieve the true \DAG{}-matrix $\semCoeffMat$, but numerical precision limits this convergence.}
\end{subfigure}

%% file: tikz/misspec_scatter.tikz
\begin{tikzpicture}[baseline]
        \pgfplotstableread[col sep=comma]{./data/misspec_scatter.csv}\datatable
    \begin{axis}[%
        ,xlabel=$\misspecCond$
        ,ylabel=$\averageCausalEffectTarget$
        ,xmin=1
        ,xmax=4
        ]
        \addplot+[only marks, mark=x, mark options={scale=0.6,color=black}] table [x=condition, y=ace_circ]{\datatable};
        \addplot[no marks, samples=2, black, dashed] {0.92};
    \end{axis}
\end{tikzpicture}

%% file: tikz/misspec_coverage_100.tikz
\begin{tikzpicture}[baseline]
    \pgfplotstableread[col sep=comma]{./data/misspec_coverage_100.csv}{\datatable};
    \begin{axis}[%
        ,xlabel=$\misspecCond$
        ,ylabel=Empirical Coverage
        ,yticklabel={\pgfmathparse{\tick*100}\pgfmathprintnumber{\pgfmathresult}\%},
        ,ymax=1.01
        ,ymin=0.7
        ,xmin=1
        ,xmax=4
        ]

        \addplot[dashed] table[x=x,y=target]{\datatable};
        \addplot[const plot mark right, thin, white, name path=A]  table [x=x, y=yUpper]{\datatable};
        \addplot[const plot mark right, thin, white, name path=B]  table [x=x, y=yLower]{\datatable};
        \addplot[black!10] fill between[of=A and B];
        \addplot[const plot mark right,thick]  table [x=x, y=y]{\datatable};
    \end{axis}
\end{tikzpicture}  

%% file: tikz/misspec_coverage_10000.tikz
\begin{tikzpicture}[baseline]
        \pgfplotstableread[col sep=comma]{./data/misspec_coverage_10000.csv}\datatable
    \begin{axis}[%
        ,xlabel=$\misspecCond$
        ,ylabel=Empirical Coverage
        ,yticklabel={\pgfmathparse{\tick*100}\pgfmathprintnumber{\pgfmathresult}\%},
        ,ymax=1.01
        ,ymin=0.8
        ,xmin=1
        ,xmax=4
        ]

        \addplot[dashed] table[x=x,y=target]{\datatable};
        \addplot[const plot mark right, thin, white, name path=A]  table [x=x, y=yUpper]{\datatable};
        \addplot[const plot mark right, thin, white, name path=B]  table [x=x, y=yLower]{\datatable};
        \addplot[black!10] fill between[of=A and B];
        \addplot[const plot mark right,thick]  table [x=x, y=y]{\datatable};
    \end{axis}
\end{tikzpicture}  